\documentclass[12pt,titlepage]{article}

\usepackage{graphicx}
\usepackage{amsmath}
\usepackage{subcaption}
\usepackage{multirow}

\newtheorem{theorem}{Theorem}

\newtheorem{corollary}[theorem]{Corollary}

\newtheorem{proposition}[theorem]{Proposition}

\newenvironment{proof}[1][Proof]{\noindent\textbf{#1.} }{\ \rule{0.5em}{0.5em}}
\input{tcilatex}
\begin{document}

\title{ Search-and-Rescue Rendezvous}
\author{Pierre Leone and Steve Alpern}

\begin{abstract}
We consider a new type of asymmetric rendezvous search problem in which
Agent II needs to give Agent I a `gift' which can be in the form of
information or material. The gift can either be transfered upon meeting, as
in traditional rendezvous, or it can be dropped off by II at a location he
passes, in the hope it will be found by I. The gift might be a water
bottle for a traveller lost in the desert; a supply cache for Lieutenant
Scott in the Antarctic; or important information (left as a gift). The
common aim of the two agents is to minimize the time taken for I to either
meet II or find the gift. We find optimal agent paths and droppoff times
when the search region is a line, the initial distance between the players
is known and one or both of the players can leave gifts. When there are no gifts this is the
classical asymmetric rendezvous problem solved by Alpern and Gal in 1995 \cite{alperngal1995}.  We exhibit strategies solving these various problems
and use a `rendezvous algorithm' to establish their optimality.

\end{abstract}

\maketitle

\section{Introduction}

The rendezvous search problem asks how two (or more) non-communicating
unit-speed agents, randomly placed in a known dark search region, can move
about so as to minimize their expected meeting time. It was first proposed
in \cite{AlpernSeminar} and formalized in \cite{doi:10.1137/S0363012993249195}. Here we take the search
region to be the real line. The form of the problem dealt with in this paper
is called the `player-asymmetric' (or 'indistinguishable player') version.
This means that the players (agents) can agree before the start of the game
which strategies each will adopt: for example if the search region were a
circle, they could agree that Player I would travel clockwise and Player II
counter-clockwise. In our problem they begin the game by being placed a
known distance $D$ apart on the line, but of course neither knows the
direction to the other player. The solution to this problem given in (Alpern
and Gal, 1995 \cite{alperngal1995}) achieves a rendezvous value (least expected meeting time) of $%
13D/8.$ Their solution is best presented as a modification of the simple
`wait for Mommy' strategy in which Player II\ (Baby) stays still and I
(Mommy) searches optimally for an immobile target: Mommy first goes distance 
$D$ in one direction (chosen randomly) and then $2D$ in the other. This
gives an expected meeting time of $\left( 1/2\right) D+\left( 1/2\right)
\left( D+2D\right) =2D.$ If Baby optimizes against Mommy's strategy, he goes
a distance $D/2$ in some direction (hoping to meet an oncoming Mommy) and
then back to his starting location at time $D$ (in case Mommy comes there
from the other direction). If he has not met Mommy, he knows she is now at a
distance $2D$ because she went in the wrong direction. So he goes a distance 
$D$ is some direction and then back to his starting location again. The
equally likely meeting times for this `modified wait for Mommy strategy' are 
$D/2,$ $D,$ $2D$ and $3D$ with the stated average of $13D/8,$ which can be
shown to be the best possible time. Two improved proofs of the optimality of
this strategy pair are given in \cite{alpern2003search}.

This paper introduces a new asymmetry into the rendezvous problem. We
consider that one of the players (taken as I) is lost and needs a `gift',
say food or water. The other player (taken as II) is the rescuer, who has
plenty of this resource to give. He can either give it to I by finding him,
as in the original rendezvous problem, or by leaving a canteen of water or
cache of food which is later found by the thirsty or hungry I. That is, the
game ends when Player I either meets Player II or finds a gift he has left
for him. This is the case where there is one gift and we denote this game $G_1$.

Another version of this problem is where each player has a gift. First
consider the case where a boy and girl who like each other are trying to
find each other at a rock concert. They can move about and hope to meet and
they can also leave their phone number on some bulletin board. The game ends
when they meet or when \textit{one of them }finds the other's phone number
and calls. More generally, the game ends when the players meet or when one
finds a gift left by the other. Here the gift is interpreted as containing
information rather than resources. We denote this game by $G_2^{or}$.  A second case of two gifts is when two
players are lost and one has food but needs water and the other has water
and needs food. One can leave a cache of food along his path while the other
can leave a canteen of water. (Each still has plenty left for himself or to
give to the other upon meeting.) Here the game ends when either the two
players meet or when \textit{both players} have found the gift left by the
other. We denote this game by $G_2^{and}$.

We solve all of these rendezvous problems. For comparison it is easiest to
take the initial distance as $D=16,$ so that the original rendezvous time
(expected time for the game to end) of Alpern-Gal is $26=13(16)/8$. This
time goes down in all cases. Table \ref{table:games} summarizes our main results,
where $\tau $ denotes the optimal time(s) to drop off the gift(s) and $\bar{R%
}$ is the rendezvous value (least expected meeting time)%
\hskip -1cm \begin{table}
\begin{eqnarray*}
&&.
\hskip -1.5cm\begin{tabular}{lllll}
Name &\vline&Problem & $\tau $ & $\bar{R}$ \\ 
\hline
G&\vline&No gifts  (Alpern-Gal, 1995) & - & 26 \\ 
$G_1$&\vline&One gift & 4 & [20.99984, 21] \\ 
$G_2^{or}$&\vline&Each player has a gift, one must be found &(8,8)  & [19.99968, 20]  \\ 
$G_2^{and}$&\vline&Each player has a gift &(0,0), (8,8), (4,x)  & [23.99968, 24] 
\end{tabular}
\\
&&
\end{eqnarray*}\label{table:sum}\caption{Main results for optimal dropoff times and rendezvous value for $D=16$. The x denotes immaterial dropping time. The best strategies found are explicited in Theorems \ref{theo:g1}, \ref{theo:g2or} and \ref{theo:g2and}.}\label{table:games}
\end{table}

When there is one gift, the rendezvous time is $21$ (the gift is dropped at
time $D/4=4$). When there are two gifts, but only one has to be found to end
the game, the rendezvous time is $20$ (the gifts are both dropped at time $D/2=8$). When there are two gifts and both must be found to end the game, the
rendezvous time is $24$, with the gifts being
dropped at times $(0,0)$, $(D/4=4,x)$ or $(D/2=8,D/2=8)$. These results are summarized in Table \ref{table:games}.

It will be of interest to consider all of these problems in the two
dimensional setting of a planar grid ($Z^{2\text{ }}),$ as initiated for
asymmetric rendezvous \cite{andersonfekete2001} and studied in \cite{chester2004}, and on arbitrary graphs as studied in \cite{alpern2002a}. The use of
gifts could also be studied in the rendezvous contexts of Howard
\cite{howard1999}, Lim \cite{lim1997}, Anderson and Essegaier \cite{andersonessegaier1995}, Han et al \cite{han2008} and in other settings discussed
in the survey Alpern \cite{alpern2002b}. Gifts  might also be used in the
discrete rendezvous problem solved by Weber \cite{weber2012}.

The main tool that we use is to confine the rendezvous strategies to a finite set  when the time (or times) of dropping the gifts is fixed, see section \ref{sec:dropknown}. Actually, in Proposition \ref{prop:noturns}  we identify a finite set of particular events that can be turning points for the two players. In between turning points, players can only move in a fixed direction at maximal velocity. Then, we compute the game values for dropping times that vary
over a grid with small diameter, and continuity properties of the
rendezvous time provide bound for the optimal solution, see section \ref{sec:boundingsol}.

The paper is divided into sections as follows. The next section is dedicated to the literature review. In section \ref{sec:introresults} we explain how
the problem can be thought of in terms of a single Player I who wants to
find all four `agents' of Player II. This type of analysis goes back to
\cite{alperngal1995}. In section \ref{sec:results} we give the solution to all of the
versions of the problem  discussed
above. Simply presenting these strategies and evaluating the possible
meeting times gives an upper bound on the rendezvous times. But the
possibility of better strategies is not ruled out. Section 5 presents our
algorithm for optimizing rendezvous strategies when the dropping-times are
given, and thus proves the approximate or exact optimality of each of the
strategies given in section 4. We show that they minimize the rendezvous
times over the dropping time grid. In section 6, we show how the optimal values of the new games can be upper and lower bounded given exact values computed for dropping times in a regular mesh. Moreover, we show how bound the regions where the dropping-times leading to optimal strategy are. These results are summarized in Theorems \ref{theo:one} and \ref{theo:two}. Sections 7 and 8 present numerical values for the lower and upper bounds. The lower bound are not computed as precisely as 'possible'. Improving the accuracy is only a matter of computing time.

\section{Literature Review}

The rendezvous search problem was first proposed by Alpern \cite{AlpernSeminar} in a
seminar given at the Institue for Advanced Study, Vienna. Many years passed
before the problems presented there were properly modeled. The first model,
where the players could only meet at a discrete set of locations,  was
analyzed by Anderson and Weber \cite{10.2307/3214827}. This difficult problem was later
solved for three locations by Weber \cite{weber2012}. Rendezvous-evasion on discrete
locations was studied by Lim \cite{lim1997} and solved for two locations (boxes) by
Gal and Howard \cite{galhoward2005}. 

The rendezvous search problem for continuous space and time, including the
infinite line, was introduced by Alpern \cite{doi:10.1137/S0363012993249195}. The player-asymmetric form
of the problem (used in this paper), where players can adopt distinct
strategies, was introduced in Alpern and Gal \cite{alperngal1995}. Baston and Gal \cite{bastongal2001}
allowed the players to leave markers at their starting points (e.g. the
parachutes they used to arrive). The last two papers form the starting point
of the present paper. We apply the same method as the one  in this paper to the rendezvous problem on the line with markers. Our findings are that the solution in \cite{bastongal2001} seems optimal even if we allow the dropping times of markers to be chosen. Results are going to be published independently of this paper.

The corresponding player-symmetric problem on the line
was developed by Anderson and Essegaier \cite{andersonessegaier1995}. Their results have been
successively improved by Baston \cite{baston1999}, Gal \cite{gal1999}, and Han et al \cite{han2008}.
These papers assumed that the initial distance between the players on the
line was known. The version where the initial distance between the players
is unknown was studied by Baston and Gal \cite{bastongal1998}, Alpern and Beck \cite{alpernbeck1999,alpernbeck2000}
 and Ozsoyeller et al \cite{ozsoyeller2013}.

The continuous rendezvous problem has also been extensively studied on
finite networks: the unit interval and circle by Howard \cite{howard1999}; arbitrary
networks by Alpern \cite{alpern2002a}; planar grids by Anderson and Fekete \cite{andersonfekete2001} and 
Chester and Tutuncu \cite{chester2004}, the star graph by Kikuta and Ruckle \cite{kikuta2007}. 

The present paper is an application of rendezvous search to
`search-and-rescue' operations. A different application of search theory to
that area is in Alpern \cite{alpern2011}, where the Searcher must find the Hider
(injured person) and then bring him back to a specified first aid location.
An application of rendezvous to robotic exploration is given in Roy and
Dudek \cite{roy2001}. An application of rendezvous to the communications problem of
finding a common channel is given in Chang et al \cite{chang2015}. Using markers in communication networks to help matching publishers and consumers of information is suggested in \cite{DBLP:conf/safecomp/LeoneM13,DBLP:conf/safecomp/LeoneM14,kundig2016}. These works have relevant applications to anonymous communication networks where the content of information is important (content based routing). It is observed that decentralized search strategies prove to be efficient in terms of congestion and the search times are well concentrated. A survey of the
rendezvous search problem is given in Alpern \cite{alpern2002b}.


\section{Formalization of the Problem(s)}\label{sec:introresults}

We begin by presenting the formalization of the problem when there are no
gifts,  as given in \cite{alperngal1995}. Two players, $I$ and $II,$
are placed a distance $D$ apart on the real line, and faced in random
directions. They are restricted to moving at unit speed, so there position,
relative to their starting point, is given by a function $f\left( t\right)
\in \mathcal{F}$ where 
\begin{equation}\label{equ:bigspace}
\mathcal{F}~=\left\{ f:\left[ 0,T\right] \rightarrow R,~f\left( 0\right)
=0,\left\vert f\left( t\right) -f\left( t^{\prime }\right) \right\vert \leq
\left\vert t-t^{\prime }\right\vert \right\} ,
\end{equation}
for some $T$ sufficiently large so that rendezvous will have taken place. In
fact optimal paths turn out to be much simpler. We will see that optimal
paths are piecewise linear with slopes $\pm 1$ and so they can be specified
by their turning points. Suppose $I$ chooses path $f\in \mathcal{F}~$and $i$
chooses path $g\in \mathcal{F}~.$ The meeting time depends on which way they
are initially facing. If they are facing each other, the meeting time is
given by%
\[
t^{1}=t^{\rightarrow \leftarrow }=\min \left\{ t:f\left( t\right) +g\left(
t\right) =D\right\} .
\]%
If they are facing away from each other, the meeting time is given by%
\[
t^{2}=t^{\leftarrow \rightarrow }=\min \left\{ t:-f\left( t\right) -g\left(
t\right) =D\right\} .
\]%
If they are facing the same way, say both left, and $I$ is on the left, the
meeting time is given by%
\[
t^{3}=t^{\leftarrow \leftarrow }=\min \left\{ t:-f\left( t\right) +g\left(
t\right) =D\right\} .
\]%
If $I$ is on the left and they are both facing right, the meeting time is
given by%
\[
t^{4}=t^{\rightarrow \rightarrow }=\min \left\{ t:+f\left( t\right) -g\left(
t\right) =D\right\} .
\]%
To summarize, the four meeting times when strategies (paths) $f$ and $g$ are
chosen are given by the four values, see Figure \ref{fig:rdvopt},%
\[
\min \left\{ t:\pm f\left( t\right) \pm g\left( t\right) =D\right\} .
\]%
\begin{figure}
\includegraphics[scale=0.4]{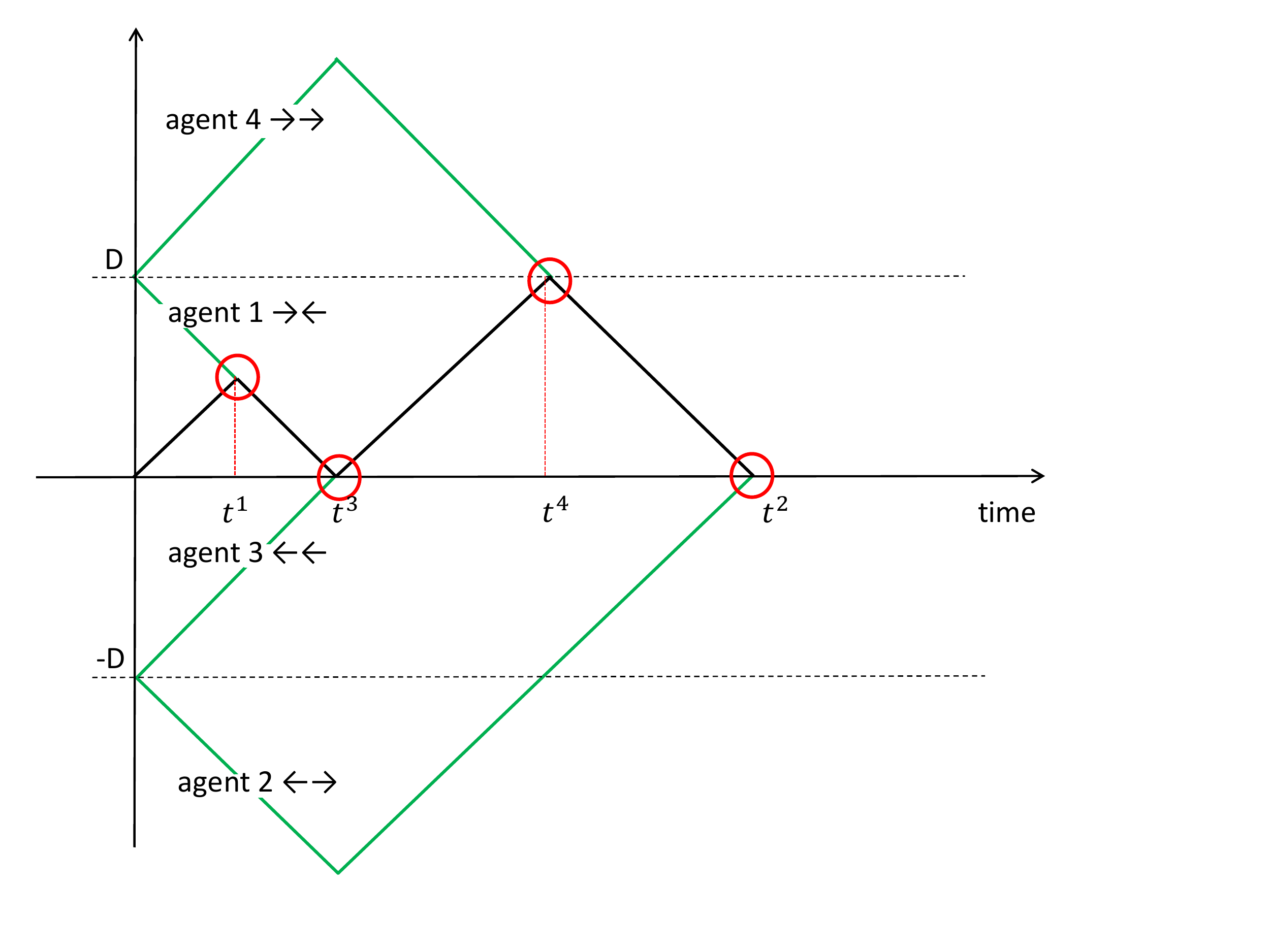}\caption{Player I starts at position $0$ and the four agents of player II representing every initial state choosen by nature.}\label{fig:rdvopt}
\end{figure}
The Rendezvous time for given strategies is their expected meeting time%
\begin{equation}
R\left( f,g\right) =\frac{1}{4}\left( t^{1}+t^{2}+t^{3}+t^{4}\right) .
\label{Rformula}
\end{equation}%
The Rendezvous Value $\bar{R}$ is the optimum expected meeting time, 
\begin{equation}
\bar{R}=\min_{f,g\in \mathcal{F}~}R\left( f,g\right) =R\left( \bar{f},\bar{g}%
\right) .  \label{Rbar}
\end{equation}

There is a simple interpretation of the formula (\ref{Rformula}) as the
average time for player $I$ (whose position at time $t$ is $f\left( t\right) 
$) to meet four agents of player $II.$ We take as the origin of the line the
starting point of Player $I$ and we take his forward direction to be the
positive direction on the line (up, if the line is depicted vertically). The
four 'agents' of $II$ start at $+D$ and $-D$ and face up or down, so their
paths are $\pm D\pm g\left( t\right) .$ The meeting times with these
`agents' are exactly the rendezvous times $t^{i},$ $i=1,2,3,4$, see Figure \ref{fig:rdvopt}. 

It has been shown for the `no gift' case, that optimal paths are of the form%
\begin{equation}\label{equ:turningpoint}
f=\left[ f_{1},\dots ,f_{k}\right] ,
\end{equation}
where the times $f_{k}$ are the turning points of the path $f,$ namely%
\[
f^{\prime }\left( t\right) =\left\{ 
\begin{array}{ll}
+1, & \text{for }f_{2j}\leq t\leq f_{2j+1}~~\text{(where }f_{0}\equiv 0\text{%
), and} \\ 
-1 & \text{for }f_{2j-1}\leq t\leq f_{2j}.%
\end{array}%
\right. 
\]%
If a player has a gift  to drop off, we denote his strategy by%
\begin{equation}\label{equ:dropturningpoint}
f=\left[ \tau ;f_{1},\dots ,f_{k}\right] ,
\end{equation}
where $\tau $ is the dropoff time and the $f_{j}$ are as above. We are now
in a position to state and illustrate the initial result of the field,  for
the case of no gifts.

\begin{theorem}[Alpern - Gal (1995b)]{\bf ($G$-Game )}
An optimal solution pair for the asymmic rendezvous problem on the line,
with initial distance $D,$ is given by, see Figure \ref{fig:rdvopt},%
\[
\bar{f}=\left[ D/2,D/2,D\right], \text{ }\bar{g}=\left[ D\right].
\]%
The corresponding meeting times are 
\begin{equation}
t_{1}=t^{1}=D/2,~t_{2}=t^{4}=D,\text{ }t_{3}=t^{3}=2D,\text{ }t_{4}=t^{2}=3D,
\label{tsubj}
\end{equation}%
with Rendezvous Value 
\[
\bar{R}=R\left( \bar{f},\bar{g}\right) =\left( D/2+D+2D+3D\right) /4=13D/8.
\]
\end{theorem}

Note that in (\ref{tsubj}) we have introduced the subscripted times $t_{j}$
as the meeting times $t^{i}$ given in increasing order. The duration of the
strategy pair is the final meeting time $t_{4}.$ We now illustrate the
optimal strategies $\bar{f},\bar{g}$ separately and then show how the
solution can be seen by drawing the single path of $I$ ($\bar{f}$) together
with the paths of the four agents of player $II$ ($\pm D\pm g\left( t\right) 
$). We take $D=2$ and draw the paths up to time $t_{4}=3D=6,$ see Figures \ref{fig:fig1alpern} and \ref{fig:fig2alpern}.

\[
\bar{f}\left( x\right) =\left\{ 
\begin{array}{lll}
x\text{ } & \text{if } & x<2 \\ 
2-\left( x-2\right)  & \text{if } & x\geq 2%
\end{array}%
\right. 
\]

\[
\bar{g}\left( t\right) =\left\{ 
\begin{array}{ccc}
t & \text{if} & t<1 \\ 
1-\left( t-1\right)  & \text{if} & 1\leq t<2 \\ 
t-2 & \text{if} & 2\leq t<4 \\ 
2-\left( t-4\right)  & \text{if} & 4\leq t\leq 6%
\end{array}%
\right. 
\]%

\begin{figure}
\includegraphics[scale=0.3]{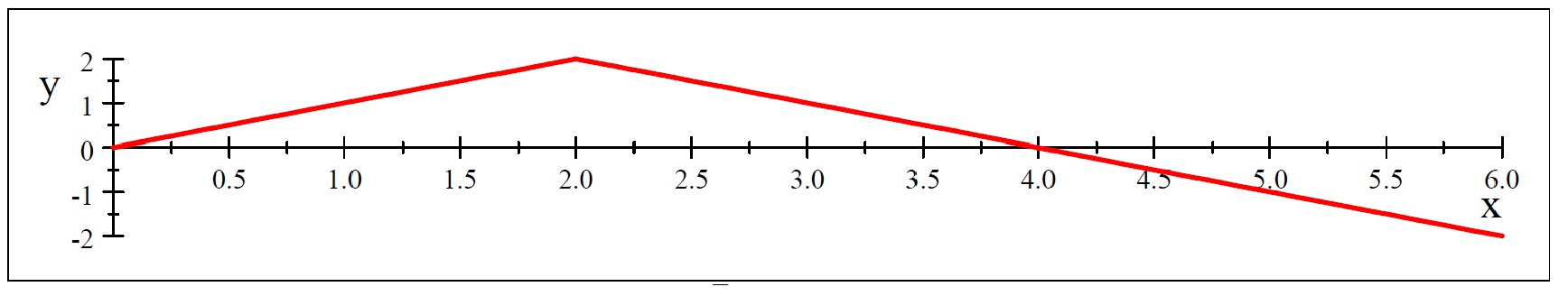}
\caption{Plot of $\bar{f}\left( t\right) $ for $D=2.$}\label{fig:fig1alpern}
\end{figure}

\begin{figure}
\includegraphics[scale=0.3]{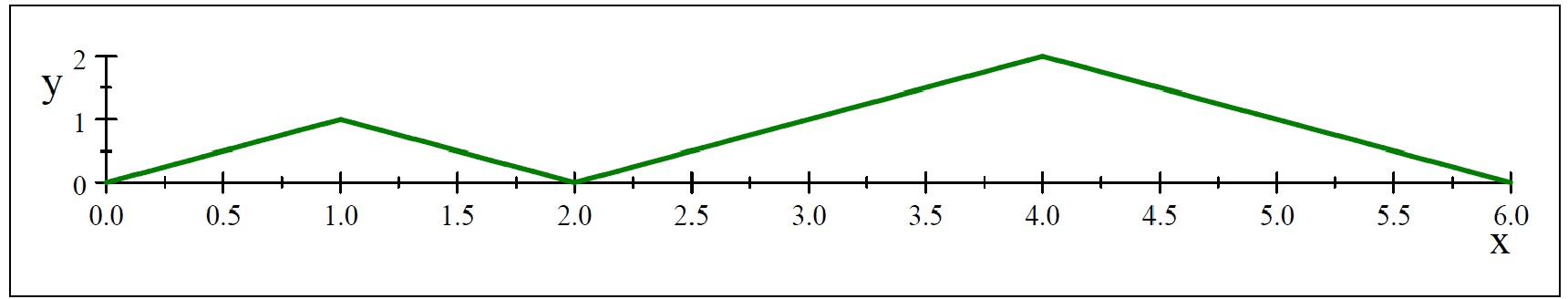}
\caption{ Plot of $\bar{g}\left( t\right) $ for $D=2.$}\label{fig:fig2alpern}
\end{figure}

\section{Upper bounds for new games' solutions}\label{sec:results}

In this section, we present the solutions of the new games that we introduce in this paper, see Table \ref{table:games}. These solutions embodies upper bounds on the rendezvous times.

\begin{theorem}\label{theo:g1} {\bf ($G_1$-game)} A solution for the asymmetric rendezvous problem on the line when one player has a gift, with initial distance $D$, is given by, see Figure \ref{fig:rdvgiftopt}
\[
\bar{f}=\left[ 3/4D\right], \text{ }\bar{g}=\left[D/4; D/4, 3/2D\right].
\]%
The corresponding times are
\begin{equation}
t_1=t^{1}=3/4D,~t_{2}=t^{4}=3/4D,\text{ }t_{3}=t^{2}=3/2D,\text{ }t_{4}=t^{3}=9/4D,
\label{tsubjgift}
\end{equation}%
with Rendezvous Value 
\[
\bar{R}=R\left( \bar{f},\bar{g}\right) =\left( 3/4D+3/4D+3/2D+9/4D\right) /4=21D/16.
\]
\end{theorem}
 
\begin{figure}
\includegraphics[scale=0.3]{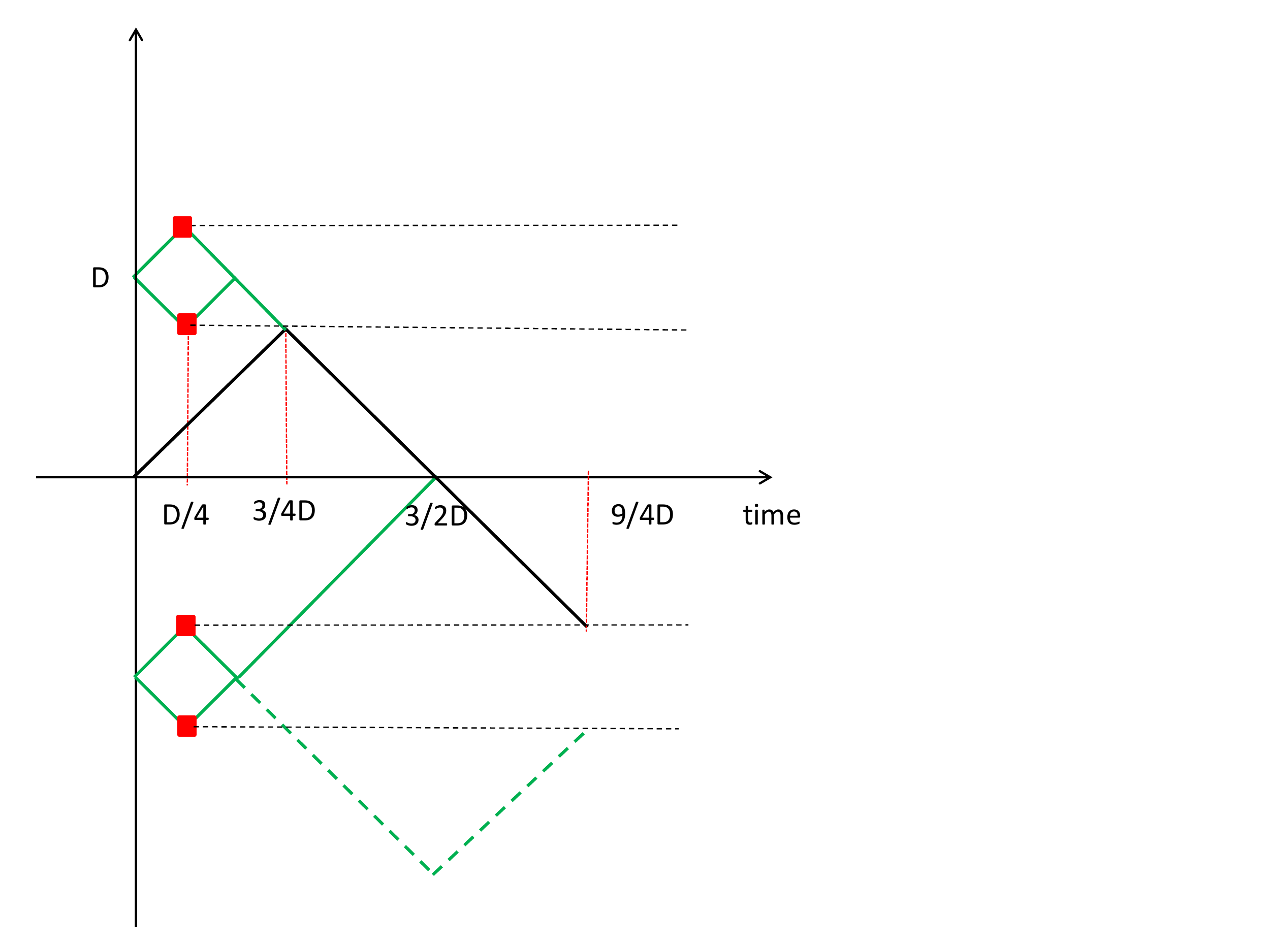}\caption{Solution of the rendezvous problem on the line with one gift ($G_1$-game). On the figure, the gift is dropped off at time $D/4$ by player II. At time $3/4D$ player I finds the gift and this ends the game with agent $4$. At the same time player I rendezvous with agent $1$. At time $3/2D$ player I rendezvous with agent $2$ and finds the gift of agent $3$ at time $9/4D$ ending the game.}\label{fig:rdvgiftopt}
\end{figure}

The next theorem present a solution of the $G_2^{or}$ game where both players have a gift. In this game, player I and agent $i$ must rendezvous or at least one of player I or agent $i$ must find the gift of the other. 

\begin{theorem}\label{theo:g2or} {\bf ($G_2^{or}$-game)} A solution for the asymmetric rendezvous problem on the line when both players have a gift and at least one must be found, with initial distance $D$, is given by, see Figure \ref{fig:rdvgiftor2opt}
\[
\bar{f}=\left[D/2; D/2\right], \text{ }\bar{g}=\left[D/2; D/2\right].
\]%
The corresponding times are
\begin{equation}
t_1=t^{1}=D/2,~t_{2}=t^{4}=3/2D,\text{ }t_{3}=t^{2}=3/2D,\text{ }t_{4}=t^{3}=3/2D,
\label{tsubjgiftor2}
\end{equation}%
with Rendezvous Value 
\[
\bar{R}=R\left( \bar{f},\bar{g}\right) =\left( D/2+3/2D+3/2D+3/2D\right) /4=20D/16.
\]
\begin{figure}
\includegraphics[scale=0.3]{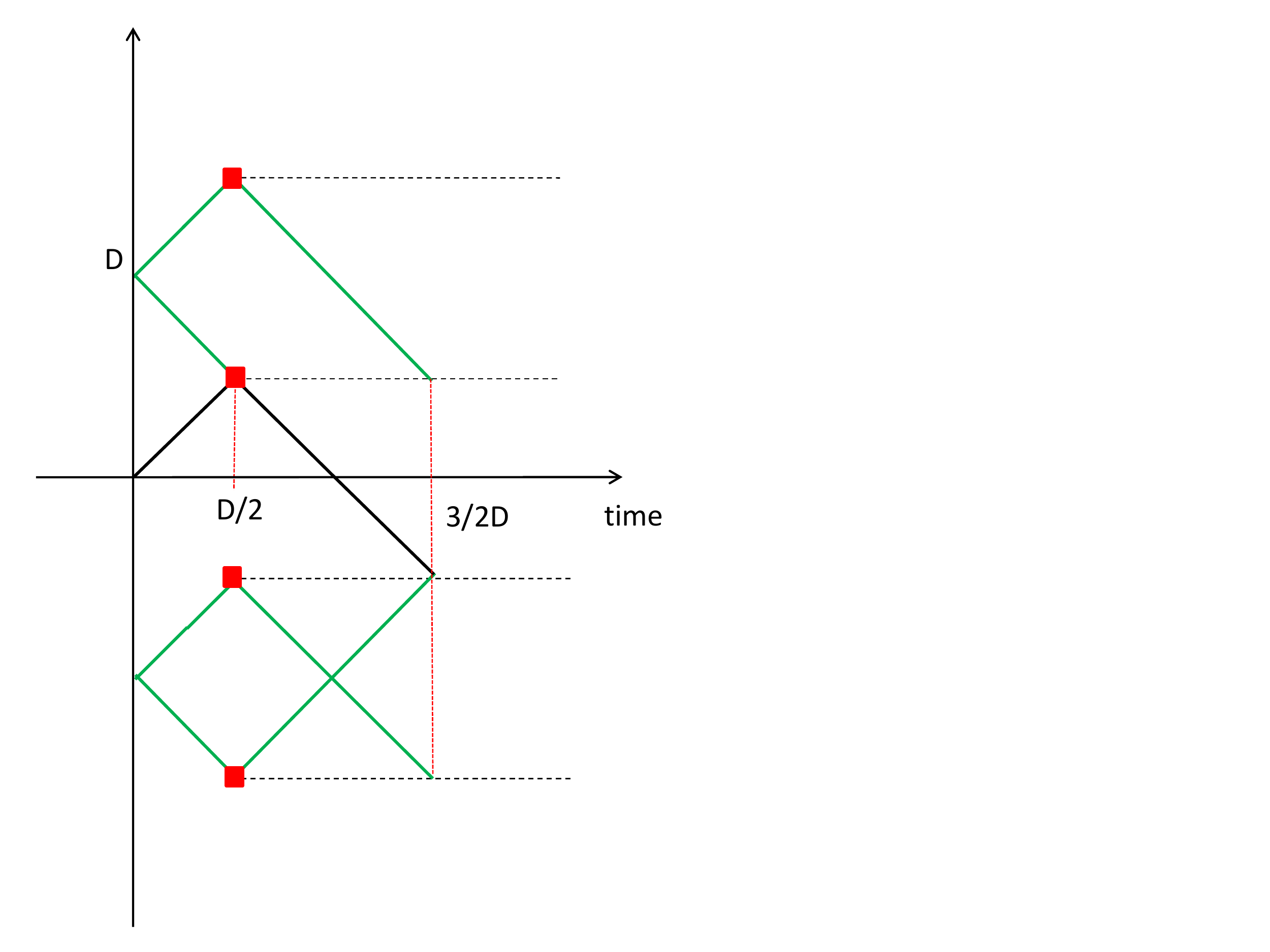}\caption{Solution of the $G_2^{or}$-game. On the figure, the gifts are dropped off at time $D/2$. At time $D/2$ player I rendezvous with agent $1$ (both find the gift simultaneously). At time $3/2D$ player I rendezvous with agent $2$, finds the gift of agent $3$ and agent $4$ finds the gift of player I.}\label{fig:rdvgiftor2opt}
\end{figure}
\end{theorem}

The next theorem present a solution of the $G_2^{and}$ game where both players have a gift. In this game, player I and agent $i$ must rendezvous or (both) must find the gift of the other.

\begin{theorem}\label{theo:g2and} {\bf ($G_2^{and}$-game)} Solution for the asymmetric rendezvous problem on the line with two gifts (one for each player), with initial distance $D$  is given by (the value $x$ denotes any time, the gift is not used), see Figure \ref{fig:marker1} as well,

\[
\bar{f_1}=\left[0;D\right], \text{ }\bar{g_1}=\left[0; D\right].
\]
\[
\bar{f_2}=\left[x;3/4D\right], \text{ }\bar{g_2}=\left[D/4; D/4, 3/4D, 7/4D\right].
\]
\[
\bar{f_3}=\left[D/2;D/2\right], \text{ }\bar{g_3}=\left[D/2; D/2, 3/2D\right].
\]
The corresponding times are
\begin{equation}
\begin{split}
&t_1=t^{1}=D/2,~t_{2}=t^{3}=3/2D,\text{ }t_{3}=t^{4}=3/2D,\text{ }t_{4}=t^{2}=5/2D,\\
&t_1=t^{4}=3/4D,~t_{2}=t^{1}=D,\text{ }t_{3}=t^{3}=7/4D,\text{ }t_{4}=t^{2}=5/2D,\\
&t_1=t^{1}=D/2,~t_{2}=t^{2}=3/2D,\text{ }t_{3}=t^{3}=2D,\text{ }t_{4}=t^{4}=2D.
\end{split}
\label{tsubjmarker1}
\end{equation}%
with Rendezvous Value 
\[
\bar{R}=R\left( \bar{f},\bar{g}\right) =24D/16.
\]
\end{theorem}

Interestingly, in this version of the game with two gifts there is a solution that makes use of only of gift, the $(\bar{f}_2. \bar{g}_2)$ strategy pair. We emphasize that the difference with the $G_1$ game is that in the $G_2^{and}$ game both players are aware that the game is finished - they both get some information or gift from the other player. In the $G_1$ game the situation is asymmetric, it may happen that only the player that gets the gift is aware of the end of the game, the information flows in only one direction. Hence, although there is a  solution  for the $G_2^{and}$ game that makes use of only one gift, the solutions to the $G_1$ and $G_2^{and}$ are very different.

\begin{figure} 
\hskip -1cm\includegraphics[scale=0.25]{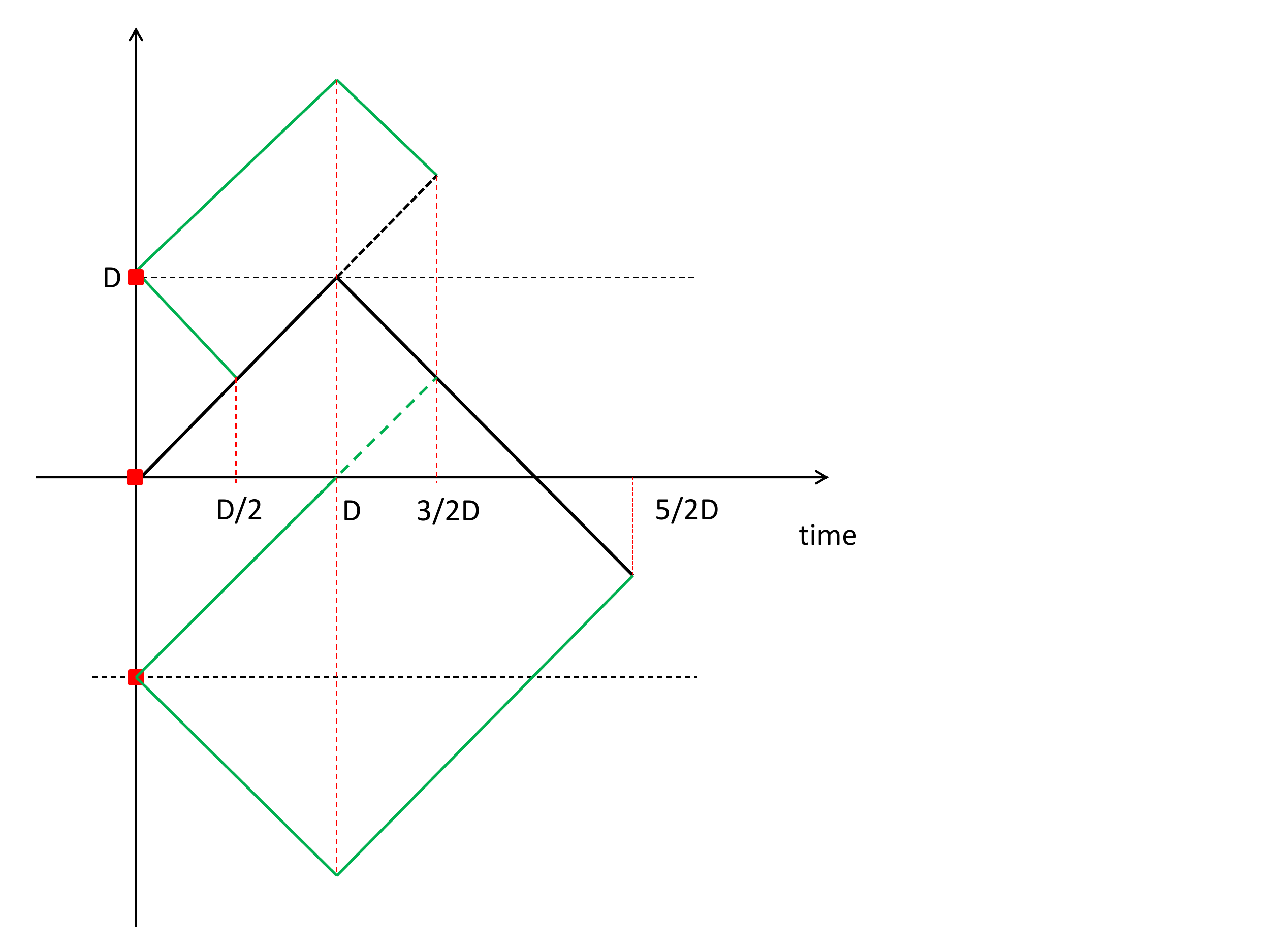}\hskip-2.cm\includegraphics[scale=0.25]{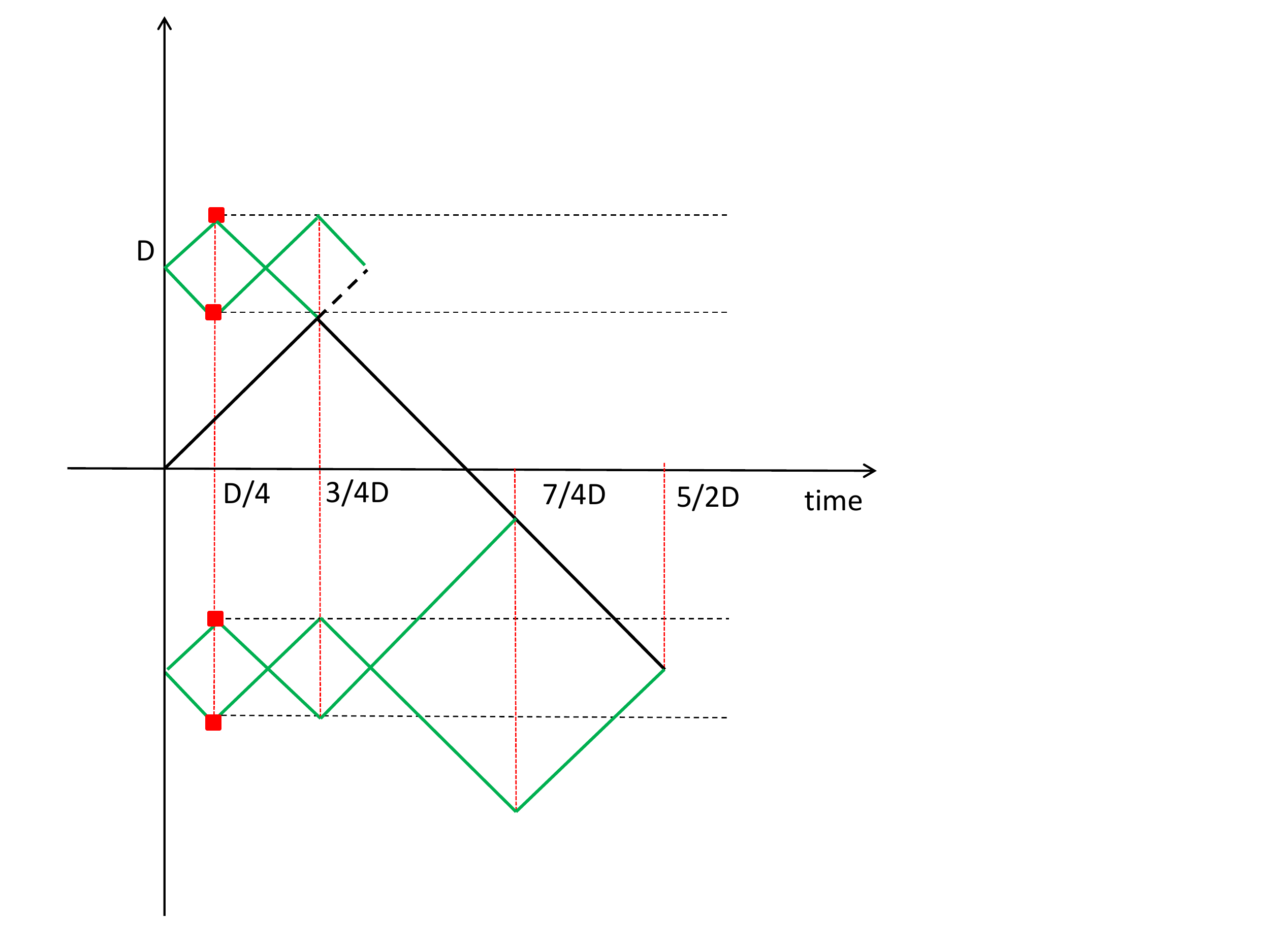}\hskip-2.5cm\includegraphics[scale=0.25]{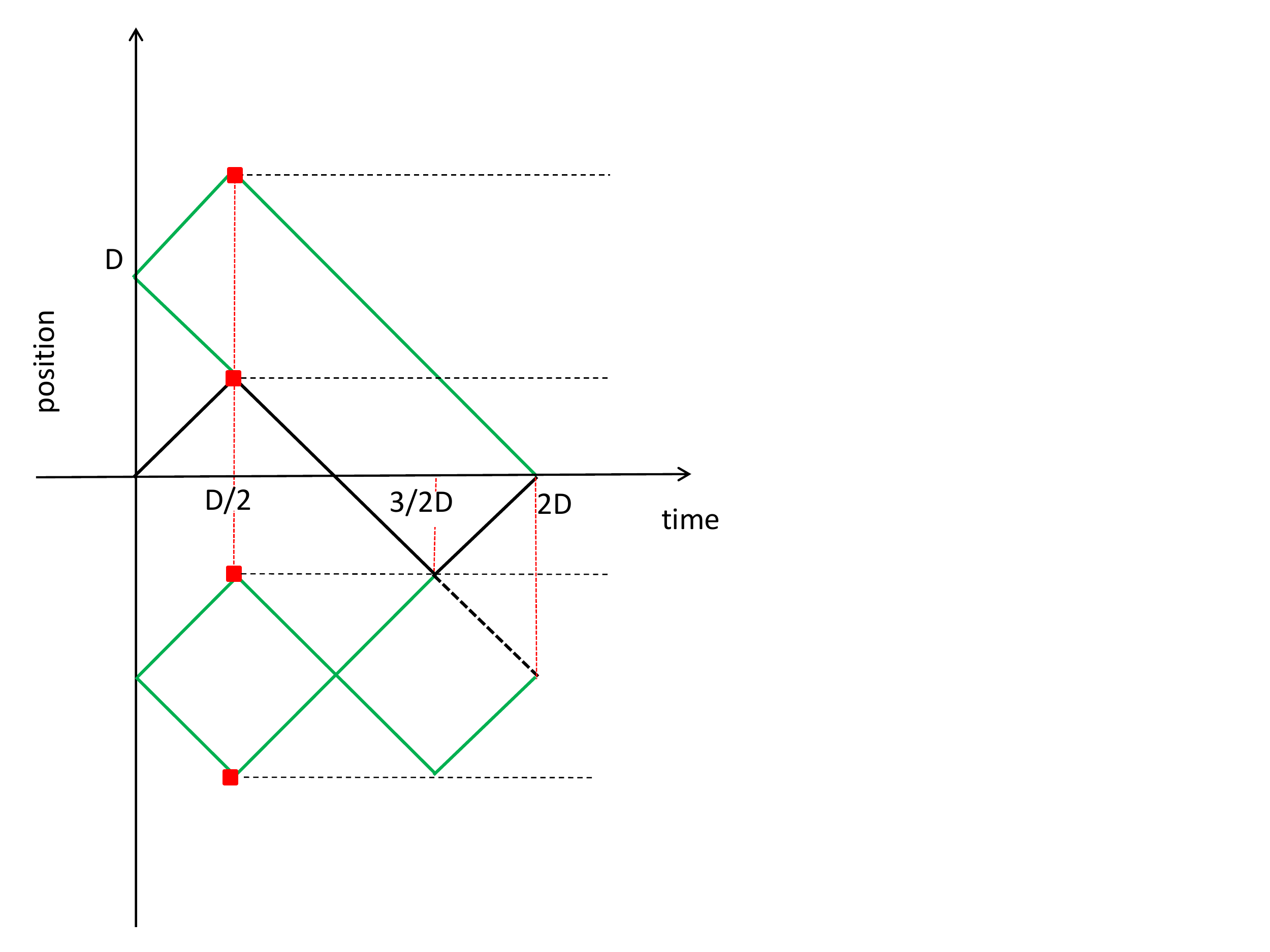}
\caption{Solutions of the $G_2^{and}$-game. From left to right the dropping times are $(0,0), (x,D/4), (D/2, D/2)$. The left solution is already known, see \cite{baston2001rendezvous} that use marks, i.e. a mark is dropped off by a player in order to indicate that he was here. Although marks seem less powerful than gifts we havn't found better solution. In the middle, the gift of player I is not used, the immaterial dropping time is denoted by $x$.}\label{fig:marker1}
\end{figure}

\section{Optimal strategy when the dropping time is known}\label{sec:dropknown}

In this section we show how computing the optimal strategy of the games $G_1, G_2^{or}, G_2^{and} $ if the dropping times are known. First observe that in all solutions presented in section \ref{sec:results} the players move at maximal speed and in direct way toward the location where he finds/drops off a gift or meets the other player. We prove in Proposition \ref{prop:noturns} that strategies that departs from this principles cannot be optimal. Hence, if we know the dropping times there are only a finite set of strategies that are candidate for being optimal (compare with the original strategy space $(\ref{equ:bigspace})$). By testing all elements of the finite set (with a program) we identify optimal strategies.

The following result is a generalization of , Lemma 5.1 of \cite{alperngal1995}
and Theorem 16.10 of  \cite{alpern2006theory}.

\begin{proposition}\label{prop:noturns}
Let $G$ be any asymmetric rendezvous game on the line where each player has at
most one gift. Then in any Nash equilibrium (NE) for $G$ (and in particular at
any optimal strategy pair) each player moves at unit speed in a fixed
direction (no turns) on each of the time intervals $J$ determined by times
$0,$ and the following times $c$:

\begin{enumerate}
\item The meeting times $c=t^{i}$ when he meets the agent $i$ of the
other player. 

\item The meeting times $c=t^{i}$ when he finds the gift dropped by agent
$i.$ 

\item The time $c=\tau$ that he drops off a gift which is later found (if he
has a gift) 

\item The times $\,c=t$ when he finds a gift dropped by agent $i$, who at a
later time $t^{i}$ finds the gift the he himself ($I$) has dropped, and for which
there is a time of types 1 - 3 later than t. (Note: this case only
occurs when both players have a gift and rendezvous requires a meeting or
that both gifts are found.)
\end{enumerate}
\end{proposition}

\begin{proof}
Assume on the contrary that for some NE strategy pair $\left(  f,g\right)  ,$
Player $I$ (say) fails the condition on some time interval $J=\left[
b,c\right]  $ of the asserted type. Suppose his path is given by $f\left(
t\right)  .$ There are three cases, depending on the what happens at time $c.$

\begin{description}
\item[1.] \textbf{At time }$c=t^{i}$\textbf{ Player I first meets agent }$i.$
Since the stated condition fails on $J,$ Player I can modify his strategy
inside the interval $J$ so that he arrives at the meeting location $f\left(
c\right)  $ at an earlier time $c-e.$ At time $c-e,$ agent $i$ of player II is
either at location $f\left(  c\right)  $ or lies in some direction (call this
$i$'s direction) from $f\left(  c\right)  .$ In the former case the meeting
with $i$ is moved forward to time $c-e.$ So Player I can stay there in until
time $c$ and then resume his original strategy, so all other meeting times are
unchanged. Otherwise, Player I goes in $i$'s direction at unit speed on
interval $\left[  c-e,c-e/2\right]  $ and then back to $f\left(  c\right)  $
at time $c,$ when he resumes his original strategy. This brings the meeting
time with $i$ no later than $c-e/2,$ without changing any other meeting times,
lowering the expected meeting time. In either case the expected rendezvous
time is lowers, contradicting the assumption that $f$ was an optimal response
to $g.$

\item[2.] \textbf{At time }$c=t^{i},$\textbf{ Player }$I$\textbf{ first finds
the gift dropped by agent }$i.~ $If the gift has just been dropped off at
time $c,$ then $I$ also meets agent $i$ at time $c,$ so the previous case
applies. Similarly, the previous case applies if the gift is not present at time $c-e$ when player I can reach the position $f(c)$.
Otherwise,  Player $I$
modifies his strategy (path) on $J$ so that he arrives at $f\left(  c\right)
$ at time $c-e$ , waits there until time $c,$ and then resumes his original
strategy $f.$ Then he finds the gift dropped by $i$ at by time $c^{\prime}$
rather than time $c,$ while all other meeting times are unchanged. This
contradicts the assumption that $f$ was a best response to $g.$

\item[3.] \textbf{At time }$c=\tau,$\textbf{ Player }$I$\textbf{ drops off a
gift found later at time }. Suppose  Player I modifies $f$ to get earlier to the dropoff location $f(c)$ at time $c-e$, drops off the gift, and then stay still until time $c$, and resumes with the original strategy $f$. After time $c-e$, case \textbf{1.} or \textbf{2.} occurs. For, if not Player I must go and meet the agent that finds the gift at a sooner time. Hence, in the interval $[c-e, t^i]$ the strategy can be further refined.This contradicts the assumption that $f$ was a best response.

\item[4.] \textbf{At time }$c$\textbf{ Player I finds a gift dropped by
agent }$i$\textbf{ who at a later time }$t^{i}$\textbf{ finds a gift dropped
by Player I. Furthermore there is a later time of type 1 - 3}. The modification of $f$ is the same as in the previous case, $\hat{f}$ simply goes
straight to $f\left(  c\right)$. If the gift is not yet present at location $f(c)$ the strategy is modified as in case \textbf{2.}, and Player I meets the agent $i$ in the interval $[c-e,c]$. Else, Player I gets the gift at time $c-e$, waits there until time $c$ and resumes with the strategy $f$. The modified strategy can modified following case 1-3 that occurs after time $c$. This contradicts that $f$ was a best response to $g$.
 
\end{description}
\end{proof}

\begin{corollary}\label{cor:noturns} Optimal strategy pair $(f,g)$ in the games $ G, G_1, G_2^{or}, G_2^{and} $ admit a representation as in $(\ref{equ:turningpoint}),(\ref{equ:dropturningpoint})$. In particular, there are only a finite set of strategies candidate for optimality.
\end{corollary}
\begin{proof}
By proposition \ref{prop:noturns} players move at full speed and turning points must coincide with location where the player drops/finds a gift   or meets the other player. The time interval between any two such events is deterministic, depending on the configuration and can be computed. If such an interval corresponds to a turning point then it leads to the definition of the $f_i$ in $(\ref{equ:turningpoint}),(\ref{equ:dropturningpoint})$. The total set of strategies in the form $(\ref{equ:turningpoint}),(\ref{equ:dropturningpoint})$ can be listed by selecting at each event either to change the direction of continue in the same direction. Because there is a finite set of such events and a finite set of directions there is a finite set of strategies in the form $(\ref{equ:turningpoint}),(\ref{equ:dropturningpoint})$.
\end{proof}

In the following we provide two examples of exact solutions whose computation is possible because of Proposition \ref{prop:noturns} (Corollary \ref{cor:noturns}).

In Figure \ref{imp:RDV} we show a procedure that enumerates all strategies in the form $(\ref{equ:turningpoint})$, solving the game $G$. The state of the system is described by the distance vector $(d_1, d_2, d_3, d_4 ,t)$ where $d_i$ is the distance for player I to agent $i$ and $t$ the current time.

 \begin{figure}
 \begin{flushleft}
{\bf procedure} CheckAllStrategy($d_1,d_2,d_3,d_4,t$)\\
$\{$\\
$\quad${\bf if} ($d_1>0$)\\
$\quad$$\quad$ print("player I : F   player II : B time : t")\\
$\quad$$\quad$CheckAllStrategy($0,d_2, d_3\widetilde   +1 ,d4,t+\frac{d_1}{2}$)\\
$\quad${\bf if} ($d_2>0$)\\
$\quad$$\quad$ print("player I : F   player II : F time : t")\\
$\quad$$\quad$CheckAllStrategy($d_1,0,d_3, d_4\widetilde  +1 ,t+\frac{d_2}{2}$)\\
$\quad${\bf if} ($d_3>0$)\\
$\quad$$\quad$ print("player I : B   player II : F time : t")\\
$\quad$$\quad$CheckAllStrategy($ d_1\widetilde  +1,d_2,0,d4,t+\frac{d_3}{2}$)\\
$\quad${\bf if} ($d_4>0$)\\
$\quad$$\quad$ print("player I : B   player II : B time : t")\\
$\quad$$\quad$CheckAllStrategy($d_1,d_2\widetilde +1,d_3,0,t+\frac{d_4}{2}$)\\
$\quad${\bf if} ($d_1+d_2+d_3+d_4==0$) \\
$\quad$$\quad$print("final time t")\\
$\}$
\end{flushleft}
 \caption{Pseudo-code of the program that enumerates candidate for optimal strategies of $G$. The recursive procedure prints out all the checked strategies and the total time. We use an operator denoted $\widetilde  +$ , i.e.  $x\widetilde + y$ which semantics is: {\bf if } $x>0$ {\bf then } $x+y$ {\bf else } 0. }\label{imp:RDV}\end{figure}

Initially $(d_1, d_2, d_3, d_4,t)=(1,1,1,1,0)$ and player I has met with agent $i$ if $d_i=0$, in which case the value of $d_i$ is frozen in the subsequent computation (this is why we introduced the operator $\widetilde +$). At the time where turning is an option by Proposition \ref{prop:noturns}, the algorithm checks whether a direction makes sense and, if yes, try it. Indeed, there are situations where going in a direction makes no sense, for instance player I cannot go forward if agents $1$ and $4$ have met already with player I. Figure \ref{eq:transition} summarizes all meaningful direction choices given the current state of the system $(d_1, d_2, d_3, d_4, t)$ as well as the new state reached after the motion. A systematic exploration of all trategies that that are eligible for optimality can be implemented by a recursive procedure, see Figure \ref{imp:RDV}. Execution of this program provides a new proof of the optimal strategy of $G$ in \cite{alperngal1995}.

The programs for solving the other games are very similar. Few more cases are allowed. For instance, for a player moving in  any direction always makes sense provided there is a player in the direction and independently of the moving direction of this player. For, player I going in the forward direction makes always sense provided agent 1 or 2 is not yet found. Actually, even if the agent goes the wrong direction, player I can still go for the gift (symmetric situations admit the same argument).

\begin{figure}
\begin{align}\label{eq:transition}
&(d_1,d_2,d_3,d_4,t) \longrightarrow_{d_1>0} (0,d_2\widetilde +d_1/2,d_3,d4,t +d_1/{2})~~\text{    direction choice : }\begin{pmatrix} F \\ F\end{pmatrix}\nonumber\\
&(d_1,d_2,d_3,d_4,t) \longrightarrow_{d_4>0} (d_1,d_2,d_3\widetilde +d_4/2,0,t+{d_4}/{2})~~\text{    direction choice : }\begin{pmatrix} F \\ B\end{pmatrix}\nonumber\\
&(d_1,d_2,d_3,d_4,t) \longrightarrow_{d_3>0} (d_1,d_2,0,d_4\widetilde +d_3/2,t+{d_3}/{2})~~\text{    direction choice : }\begin{pmatrix} B \\ F\end{pmatrix}\nonumber\\
&(d_1,d_2,d_3,d_4,t) \longrightarrow_{d_2>0} (d_1\widetilde +d_2/2,0,d_3,d_4,t+{d_2}/{2})~~\text{    direction choice : }\begin{pmatrix} B \\ B\end{pmatrix}\nonumber
\end{align}\caption{State transitions. On the right the correspondig direction choices, the upper entry of the matrix is player I and the lower is player II. The operator $\widetilde +$ semantics is:$\text{ \bf if } x>0 \text{ \bf then } x+y \text{ \bf else } 0$. }\label{eq:transition}
\end{figure}

\begin{figure}
\includegraphics[scale=0.3]{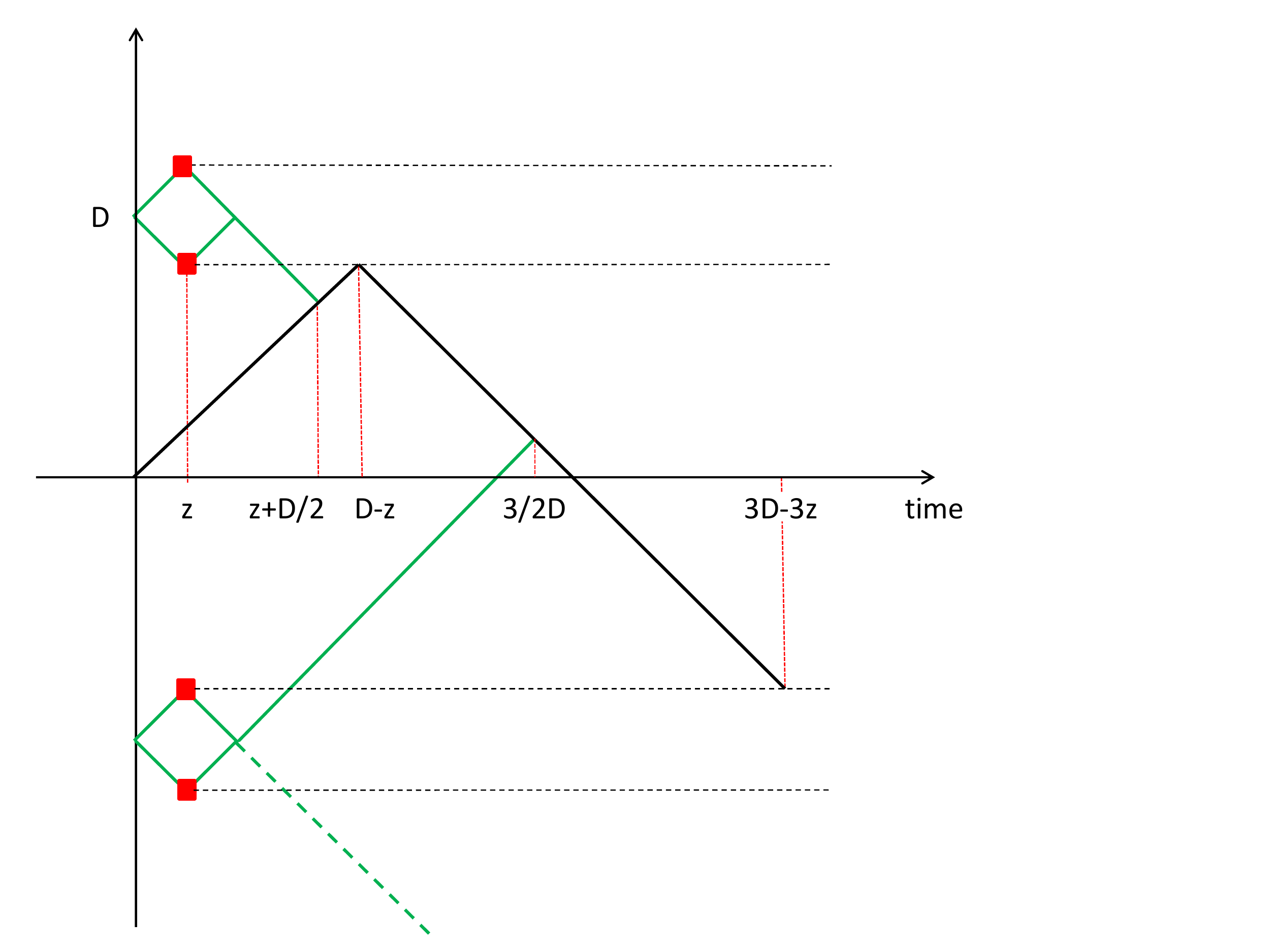}\caption{Optimal solution in $G_1$ when the dropping time is $z<1/4$. Optimal strategies are $\bar{f}=[D-z]$, $\bar{g}=[z;z]$}\label{fig:optz}
\end{figure}

To conclude this section we can illustrate how to enumerate all solution by considering the game $G_1$ when the dropping time is $z<1/4$, see Figure \ref{fig:optz}. At the onset player I and II are always going in the forward direction by symmetry. Because $z<1/4$, player II drops off the gift before any other event. Hence after time $z$ there are two strategies for player II, either it is a turning point either not.  By direct inspection with a program we checked that not turning at time $z$ is not optimal. Player I meets agent 4 at time $z+d/2$ and now both player I and II are allowed to change direction. We again checked by direct inspection that the best strategy is that player I turns and player II not. Continuing in this way, considering turning points only when this is allowed by Proposition \ref{prop:noturns} we are able to generate all strategies and select the optimal one. In this case the best strategy for player I is $[D-z]$ while for player II it is $[z;z]$.


\section{Bounding the optimal solutions}\label{sec:boundingsol}

Given the algorithm described in section \ref{sec:dropknown} a way a obtaning reasonnable knowledge about the optimal solution of one of the games is to define a mesh of values that corresponds to the dropping times. For each dropping time we compute the exact optimal solution of the game. Computing the minimal values leads to an upper bound. In the following we show how the game values that are not computed (the dropping time does not belongs to the mesh) can be bounded.

For the game $G_1$, we denote $x(l)$ the optimal game value computed by direct inspection when the dropping time is $l$. Proposition \ref{boundingRDVLetter} shows that $x(l)$ and $x(l+\alpha)$ and related by $x(l)\ge x(l+\alpha)-\alpha,~\alpha\ge 0$. Hence, if the value $x(l+\alpha)$ is computed it can be used to bound the exact solution for $x(l)$ and, if this bound is larger that upper bound computed (see Section \ref{sec:dropknown} ) the optimal solution cannot be with dropping time in the interval $[l, l+\alpha]$. This result is stated in Theorem \ref{theo:one}. Proposition \ref{continuity1} shows that the function $x(l)$ is indeed continuous, see Figure \ref{fig:giftfig}.

For the games $G_2^{or}, G_2^{and}$ the derivation of the bound is similar but requires few extra work. The result is stated in Theorem \ref{theo:two}. The analogue of Proposition \ref{boundingRDVLetter} is Proposition \ref{boundingRDVLetters}. In this case there are two dropping times and delaying only one dropping time without modifying the other is easy only if the delayed dropping time is the latest. This is the reason why in Proposition \ref{boundingRDVLetters} we distinguish $l_1\ge l_2$ or $l_1\le l_2$. Corollaries \ref{cor:bound1}, \ref{cor:bound2} and \ref{cor:bound3} specialize the results of Proposition \ref{boundingRDVLetters} to identify the region of dropping times that cannot lead to optimal solution. A direct application of these corollaries leads to Theorem \ref{theo:two}, this is illustrated on Figure \ref{fig:excluded}.

\begin{proposition}\label{boundingRDVLetter}(Bounds in $G_1$) We denote $x(l)$ the value of $G_1$ if the gift is dropped off at time $l$. For $\alpha>0$ we have that
\begin{equation}\label{eq:boundRDVLetter}
x(l)\ge x(l+\alpha)-\alpha.
\end{equation}
\end{proposition}
\begin{proof}
Let be given an optimal strategy with value $x(l)$. Considering the same strategy but instead of starting at time $0$, players are still until time $\alpha$. The resulting strategy leads to the estimate $(\ref{eq:boundRDVLetter})$. 
\end{proof}

\begin{proposition}\label{boundingRDVLetter2}(Bounds in $G_1$)\label{boundingRDVLetterminus} We denote $x(l)$ the value of $G_1$  if the gift is dropped off at time $l$. For $\alpha>0$ we have that
\begin{equation}\label{eq:boundRDVLetterminus}
x(l)\ge x(l-\alpha)-2\alpha.
\end{equation}
\end{proposition}
\begin{proof} Let $x(l)$ the optimal strategy if the gift is dropped off at time $l$. The modified strategy follows strategy $x(l)$ until time $l-\alpha$ at which it drops thegift. If the gift (or absence)  is never used by the strategy this does not change the remaining rendezvous time and hence the modified strategy statisfies $(\ref{eq:boundRDVLetterminus})$. 

 Else, write $t$ the smallest time at which an agent of player II finds the gift or notice the absence. The modified strategy follows $x(l)$ until time $t$. If the agent of player II finds the gift before time $t$ the modified strategy continues as $x(l)$. Indeed, smaller rendezvous time are consistent with $(\ref{eq:boundRDVLetterminus})$.

If the agent of player II does not find the gift before time $t$ because of the new position of the gift then, player I strategy is stopped at time $t$ until time $t+2\alpha$. In the meantime, the agent of player II continues for $\alpha$ unit of time and then back to the same position by time $t+2\alpha$. At this time, the positions of the players are the same as the one at time $t$ but now the agent of player II knows whether the gift was dropped off or not. We are then in the same situation as with the strategy $x(l)$ but delayed by $2\alpha$. Since the reasoning can be applied at most four times, we obtain the bound $(\ref{eq:boundRDVLetterminus})$.
\end{proof}

\begin{proposition}\label{continuity1} For $G_1$ the function $x(l)$ is continuous.
\end{proposition}
\begin{proof}
By combining Propositions \ref{boundingRDVLetter} and \ref{boundingRDVLetter2} we obtain that $\mid x(l+\alpha)-x(l)\mid \le \alpha$.

\begin{figure}\label{fig:continuity1}
\centering\hskip -5.8cm
\begin{subfigure}[b]{0.18\textwidth}
\includegraphics[scale=0.18]{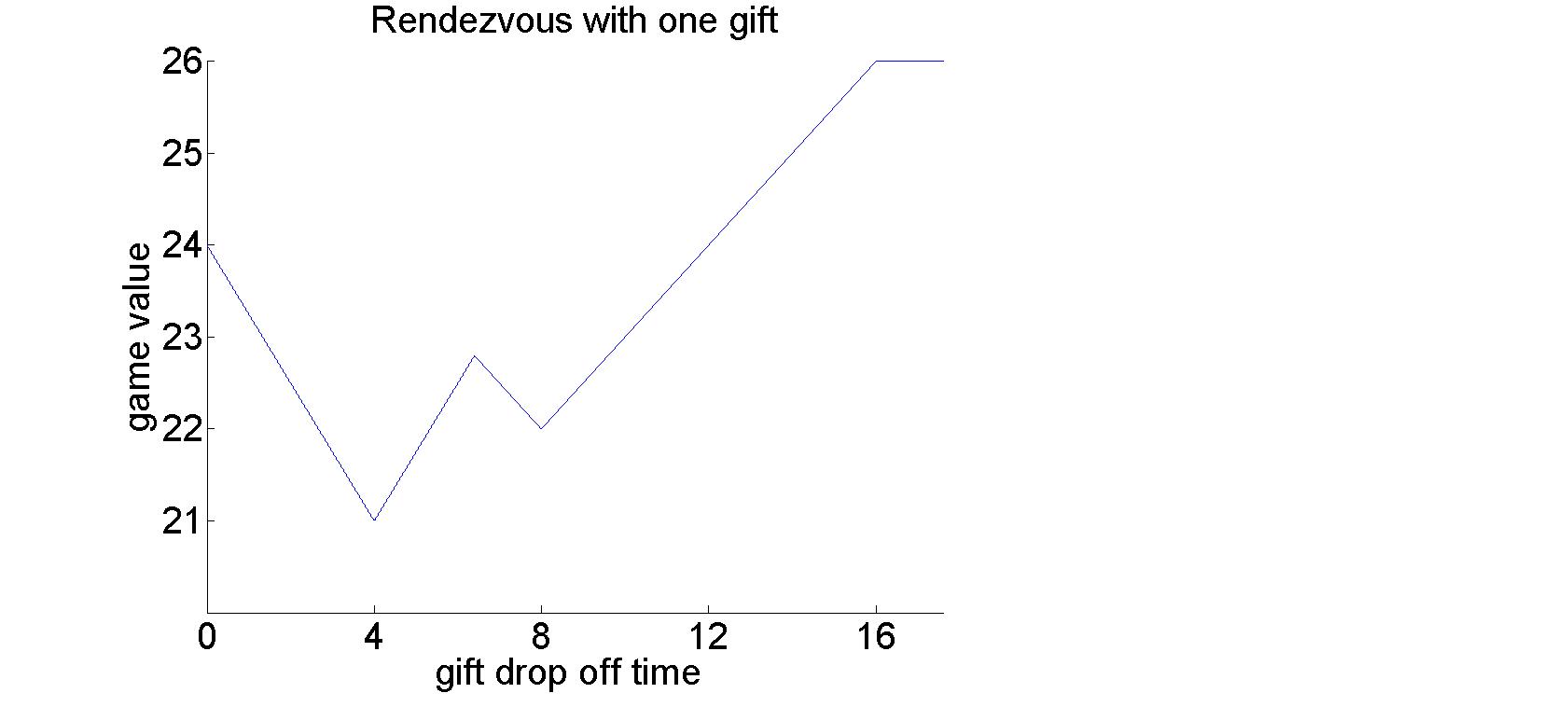}
\end{subfigure}\hskip 4.3cm 
\caption{Value of the game $G_1$  computed at dropping times in a regular mesh.}\label{fig:giftfig}
\end{figure}

\end{proof}

\begin{theorem}\label{theo:one}(Bounding optimal solutions for $G_1$) Let be given a regular mesh $\{d_i\}_{i=0,\ldots,N}$ of dropping time and the corresponding optimal solutions of $G_1$ when the dropping times are in the mesh, we denote these values $\{G_1(d_i)\}_{i=0,\ldots,N}$. Then, the optimal value of $G_1$ given that the dropping is constrained to belong to $[d_0,d_N]$ is in the interval $[x_{min}-\alpha, x_{min}]$, where $\alpha=d_{i+1}-d_i$ and $x_{min}=\min_{i}\bigl\{ G_1(d_i)\bigr\}$. Moreover, the optimal dropping times are contained in the set
\begin{equation*}
\biggl\{[d_{i-1},d_i]~\bigl|~ G(d_i)-\alpha \le x_{min}\biggr\}
\end{equation*}
  
\end{theorem}

\begin{proposition}\label{boundingRDVLetters}(Bounds for $G_2^{or}$, $G_2^{and}$) We denote $x(l_{1},l_2)$ the value of the game if the gifts are dropped off at times $l_1$ and $l_2$. For $\alpha>0$ we have that
\begin{align}
& x(l_1, l_2)\ge x(l_1+\alpha, l_2)-\alpha~~\text{if } l_1\ge l_2,\nonumber\\
&x(l_1, l_2)\ge x(l_1+\alpha, l_2+\alpha)-\alpha\label{eq:boundRDVLetters}\\
&x(l_1, l_2)\ge x(l_1, l_2+\alpha) -\alpha~~\text{if } l_2\ge l_1\nonumber
\end{align}
\end{proposition}
\begin{proof}
The idea of the proof of this proposition is analogous to the one of propositions \ref{boundingRDVLetter} and \ref{boundingRDVLetterminus}. Consider the first line of $(\ref{eq:boundRDVLetters})$ and the strategy that leads to $x(l_1, l_2)$. If $l_1\ge l_2$ we keep this strategy unchanged until time $l_1$, then we freeze the motions of both players for a time span of $\alpha$ and, finally resume the motion. In the second strategy the second player drops off the gift at time $l_1+\alpha$ and all the subsequent meeting time are delayed at most by $\alpha$ units of time. This leads to the first line of $(\ref{eq:boundRDVLetters})$. The following lines are proved in a similar way.
\end{proof}

\begin{corollary}\label{cor:bound1} Assume $l_1\ge l_2+\alpha$, if $x(l_1+\alpha,l_2)\ge x_{min} + 2\alpha$ then, $x(l_1-\beta,l_2-\gamma)\ge x_{min}$ $\forall~ 0\le\beta, \gamma \le \alpha$.
\end{corollary}
\begin{proof}
We get using the two first lines of $(\ref{eq:boundRDVLetters})$
\begin{align*}
x(l_1-\beta,l_2-\gamma)&\ge x(l_1-\beta+\gamma,l_2)-\gamma\ge x(l_1+\alpha,l_2)-\gamma-(\alpha+\beta-\gamma)\\&\ge x(l_1+\alpha,l_2)-\alpha-\beta\ge  x_{min}.
\end{align*}
The first inequality comes by adding $\gamma$ to $l_1$ and $l_2$ and the middle inequality of $(\ref{eq:boundRDVLetters})$. By assumption we have $-\alpha\le -\beta+\gamma$, hence $l_1-\beta+\gamma\ge l_2+\alpha-\beta+\gamma\ge l_2$, and the first line of $(\ref{eq:boundRDVLetters})$ applies leading to the second inequality. To conclude we use the hypothesis and direct computations.
\end{proof}

\begin{corollary}\label{cor:bound2} Assume $l_2\ge l_1+\alpha$, if $x(l_1,l_2+\alpha)\ge x_{min} + 2\alpha$ the, $x(l_1-\beta,l_2-\gamma)\ge x_{min}$ $\forall 0\le\beta, \gamma \le \alpha$.
\end{corollary}
\begin{proof} We get using the two last lines of $(\ref{eq:boundRDVLetters})$
\begin{align*}
x(l_1-\beta,l_2-\gamma)&\ge x(l_1,l_2+\beta-\gamma)-\beta\ge x(l_1,l_2+\alpha)-\beta-(\alpha-\beta+\gamma)\\&\ge x(l_1,l_2+\alpha)-\alpha-\gamma\ge  x_{min}.
\end{align*}
\end{proof}

\begin{corollary}\label{cor:bound3} Assume $l_1=l_2$, if $\min\bigl(x(l_1+\alpha,l_2), x(l,l+\alpha)\bigr)\ge x_{min}+2\alpha$ then, $x(l_1-\beta, l_2-\gamma)\ge x_{min}, \forall 0\le \beta, \gamma\le \alpha$.
\end{corollary}
\begin{proof} We assume that $\beta\ge\gamma$, the symmetric case is proved similarly. We note $l=l_1=l_2$. Then, using successively the second and third line of $(\ref{eq:boundRDVLetters})$ we obtain,
\begin{align*}
x(l-\beta,l-\gamma)&\ge x( l, l+\beta-\gamma) -\beta\ge x(l_1,l_2+\alpha)-\beta-(\alpha-\beta+\gamma)\\
&\ge x(l_1,l_2+\alpha)-\alpha-\gamma\ge x_{min}+\alpha-\gamma\ge x_{min}.
\end{align*}
\end{proof}

\begin{figure}\label{fig:cor}
\centering
\includegraphics[scale=0.3]{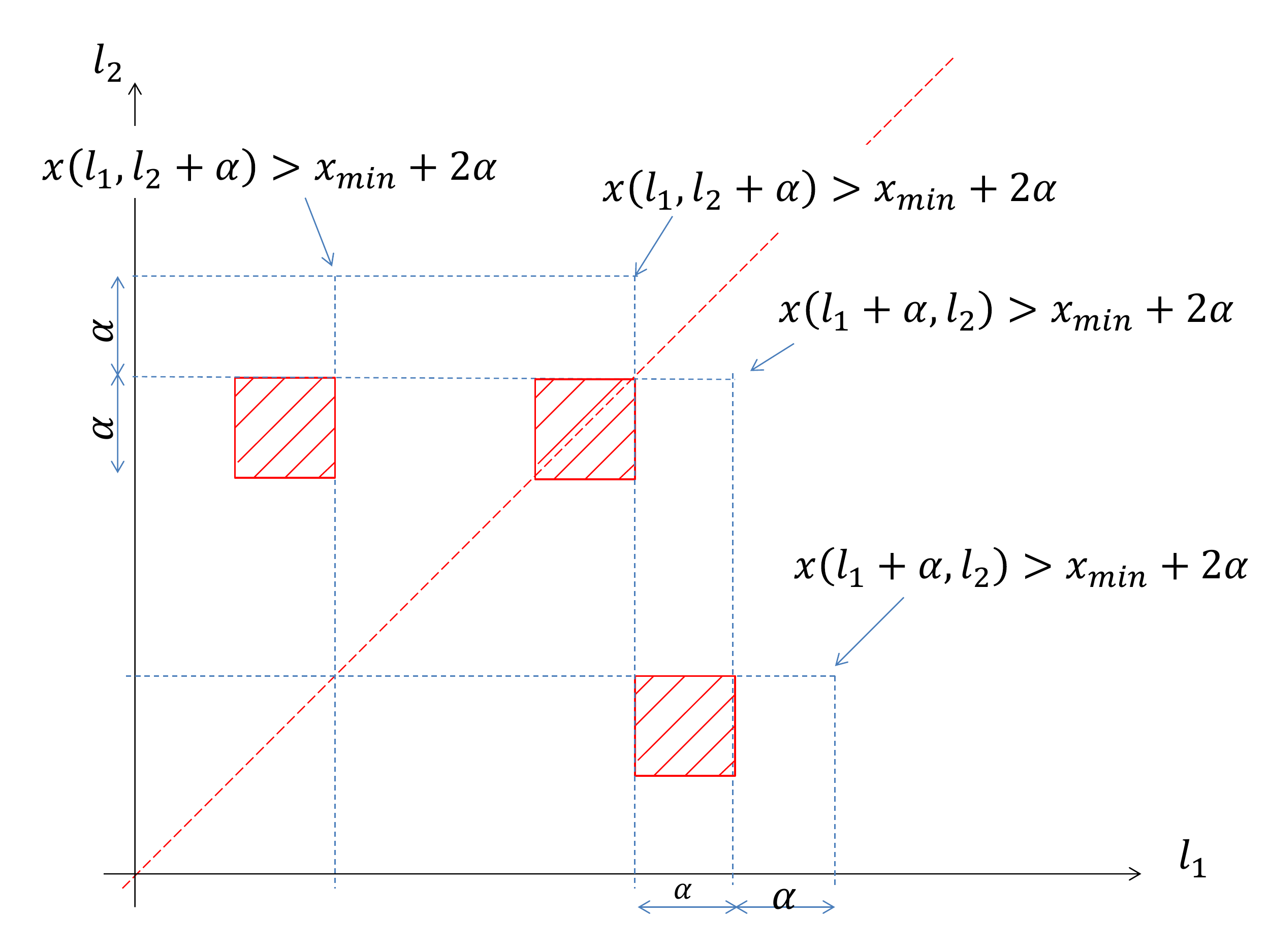}\caption{Illustrations of Corolaries \ref{cor:bound1},\ref{cor:bound2},\ref{cor:bound3}. The hatched regions cannot include dropping times leading to optimal solution.} \label{fig:excluded}
\end{figure}

\begin{theorem}\label{theo:two}(Bounds for $G_2^{or}$, $G_2^{and}$) Let be given a regular mesh $\{(d_i,d_j)\}$ for $i,j=0,\ldots,N$ of dropping times and the corresponding optimal solutions of the game $\{X(d_i,d_j)\}_{i,j=0,\ldots,N}$, where $X=G_2^{or}$  or $G_2^{and} $. Then, the optimal value of $X$ given that the dropping times belong to $[d_0,d_N]\times[d_0,d_N]$ is in the interval $\bigl[x_{min}-2\alpha, x_{min}\}\bigr]$, where $x_{min}=\min_{i,j} \{X(d_i,d_j)\}$. Moreover, the optimal dropping times are contained in the set
\begin{align*}
&\biggl\{[d_{i-1}, d_i]\times [d_{j-1}, d_j] \bigl| (i\ge j)\text{ and } X(d_{i+1}, d_j)\le x_{min} +2\alpha\biggr\} \bigcup\\
&\biggl\{[d_{i-1}, d_i]\times [d_{j-1}, d_j] \bigl| (j\ge i)\text{ and } X(d_{i}, d_{j+1})\le x_{min} +2\alpha\biggr\}\\
\end {align*}
\end{theorem}

\section{Application of Theorem \ref{theo:one}}
We show results obtained by applying Theorem \ref{theo:one} to the game $G_1$. Figure \ref{fig:giftfig} shows the plots of the game values for various dropping times. The minimal computed values are obtained for dropping time $4.0$ and the minimal values is $21$. We show in tables \ref{table:G1val} the values computed on the regular mesh aroung the better dropping times observed. Application of Theorem \ref{theo:one} shows that the optimal dropping time must be in the intervals $[3.99968, 4.00016]$. The optimal value being in the interval $[20.99984, 21]$.

\begin{table}
\begin{tabular}{ c  c  r  c    c  c  c }
Dropping time  &\vline& 3.99968 &3.99984 & 4 & 4.00016 & 4.00032 \\  \hline
$G_1$ &\vline &21.00024 & 21.00012 & 21 &21.00012 & 21.00024 \\ 
\end{tabular}\caption{The minimal solutions of $G_1$ computed with a regular mesh of $0.00016$ of the interval $[0,160]$ is $21$ for a gift dropping time of $4$. Optimal dropping times can be only around $4$, see Figure \ref{fig:giftfig}. More precisely in the interval $[3.99968, 4.00016]$ }\label{table:G1val}
\end{table}



\section{Application of Theorem \ref{theo:two}}


\begin{figure}
\hskip -2cm\includegraphics[height=7cm,width=10cm]{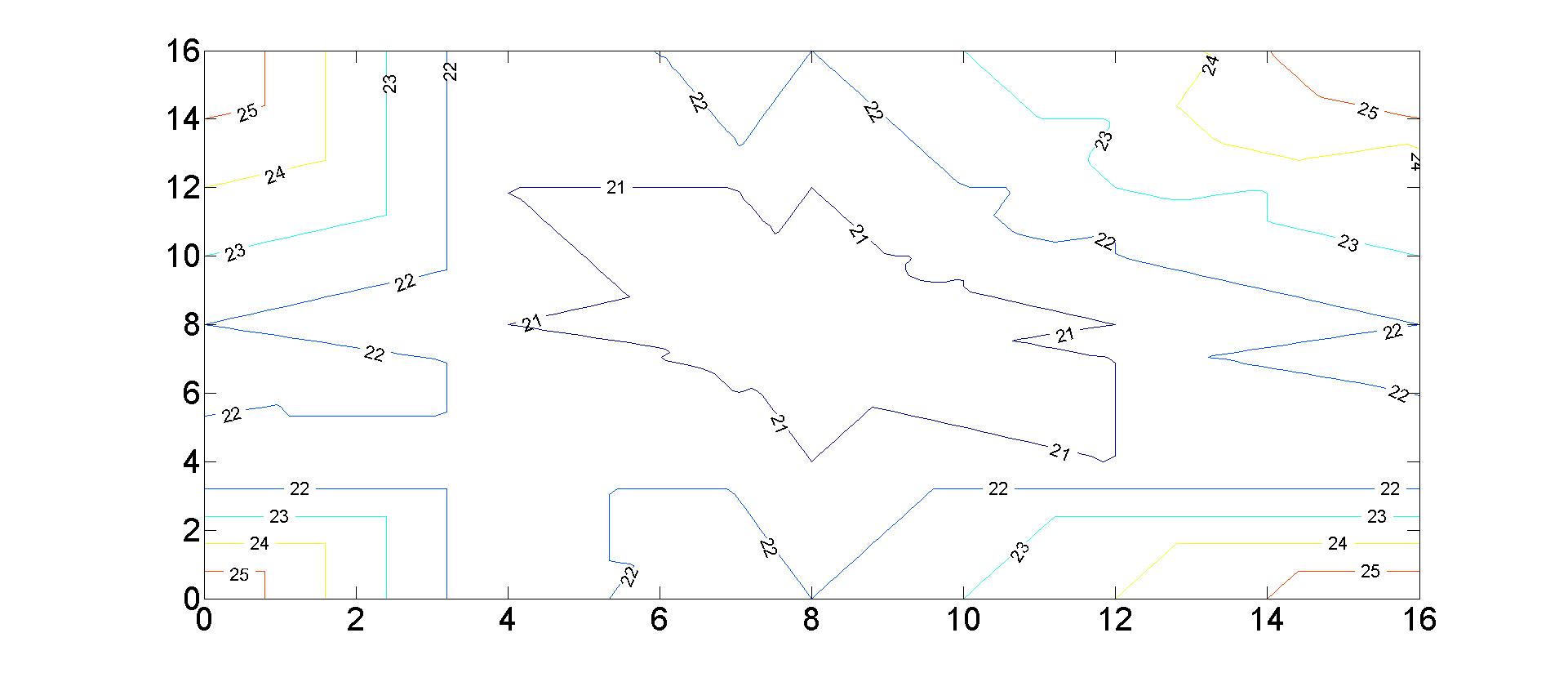}\hskip 0cm\includegraphics[height=6cm,width=8cm]{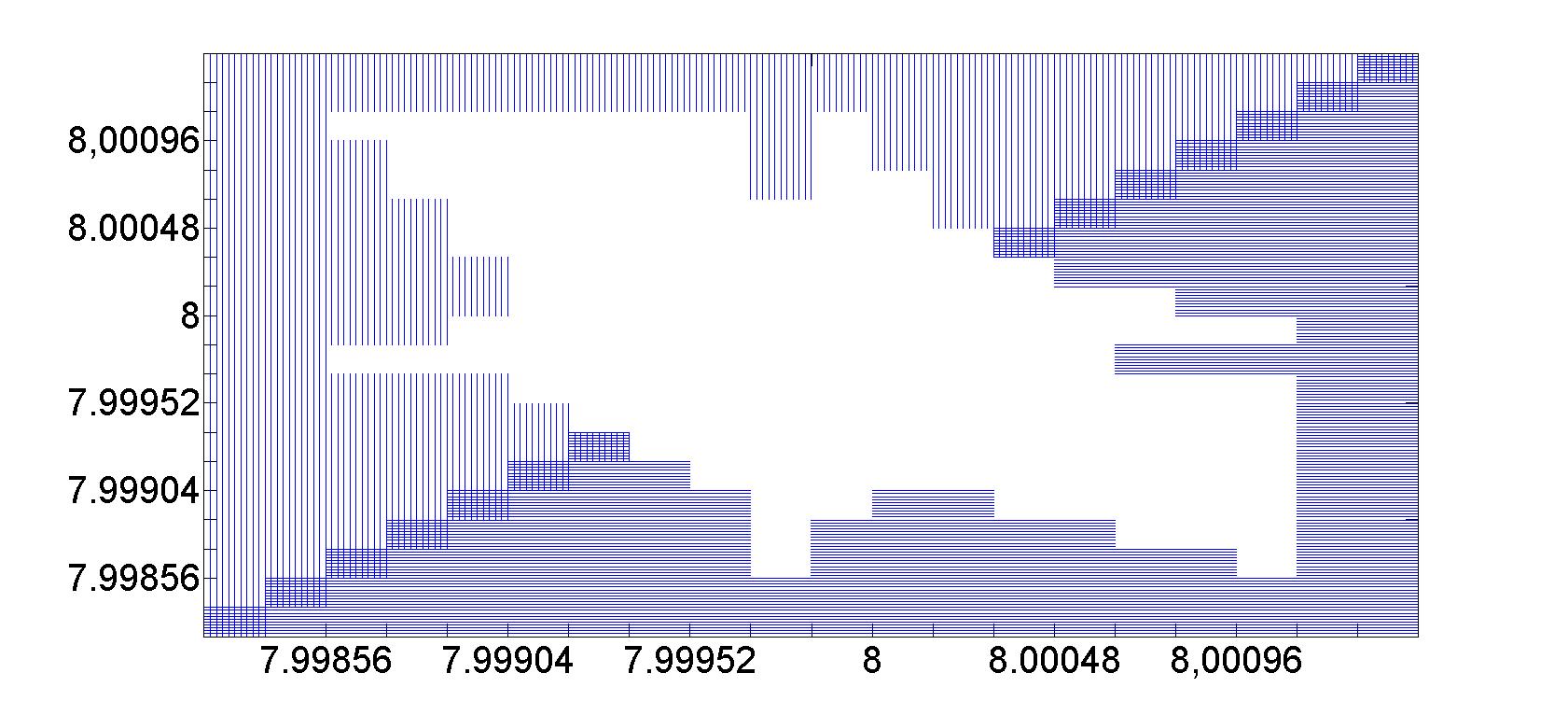}\caption{Left: Contour plot of the values of the $G_2^{or}$ game. The minimal value computed is $20$ for dropping times $ (8,8)$. Right: Enlargement of the $G_2^{or}$ game value around the dropping time $(8,8)$. The hatched region corresponds to the excluded region by Theorem \ref{theo:two}. }\label{fig:giftor}
\end{figure}
We start by applying Theorem \ref{theo:two} to $G_2^{and}$. On Figure \ref{2marksfigcontour} we observe minimal values around the marker dropping times $(0,0), (8,8), (4,x)$ and $(x,4)$. The game value is  in the interval $[23.99968,24]$. This result from the facts that the minimal value computed is $24$ and the mesh size is $0.00016$ and Theorem \ref{theo:two}.

On Figure  \ref{fig:marker()} we plot enlargement of Figure \ref{2marksfigcontour} around the points $(0,0), (4,4)$ and $(8,8)$, the hatched regions correspond to regions where it is excluded that a better solution than $24$ can be found, see Figure \ref{fig:excluded}.
Around the dropping times $(0,0)$ we see that the optimal dropping times must be in $[0, 0.00032]\times [0, 0.00016]\cup[0, 0.00016]\times [0, 0.00032]$. The values are displayed on Table \ref{table:M2val}. We have stripes of solutions corresponding to the dropping times $(4,x)$ and $(x,4)$. These stripes are contained in the region $[3.99904, 4.00048]\times x$ and $x\times [3.99904, 4.00048]$. And finally, around the dropping time $(8,8)$ the optimal ones belong to the region $[7.99856, 8.00112]\times [7.99856, 8.00112]$.

\begin{table}
\begin{tabular}{c  c  c  c  c  c  c c}
\multirow{4}{*}{$l_2$}&0.00064&\vline &24.00032&24.00040&24.00048&24.00056&24.00064\\
&0.00048&\vline &24.00024&24.00032&24.00040&24.00048&24.00056\\
&0.00032&\vline  & 24.00016 & 24.00024& 24.00032 & 24.00040&24.00048\\
&0.00016&\vline  & 24.00008 & 24.00016 & 24.00024 & 24.00032 &24.00040\\
 &0&\vline  & 24 & 24.00008 & 20.00016 & 24.00024&24.00032\\ \cline{2-8}
&&\vline  & 0 & 0.00016 & 0.00032 & 0.00048&0.00064\\ 
&&&\multicolumn{4}{c}{$l_1$}\\ 
\end{tabular}\caption{Game values for $G_2^{and}$ computed around the minimal solution found for marker dropping times $(0,0)$.}\label{table:M2val}
\end{table}

\begin{figure}
\includegraphics[scale=0.25]{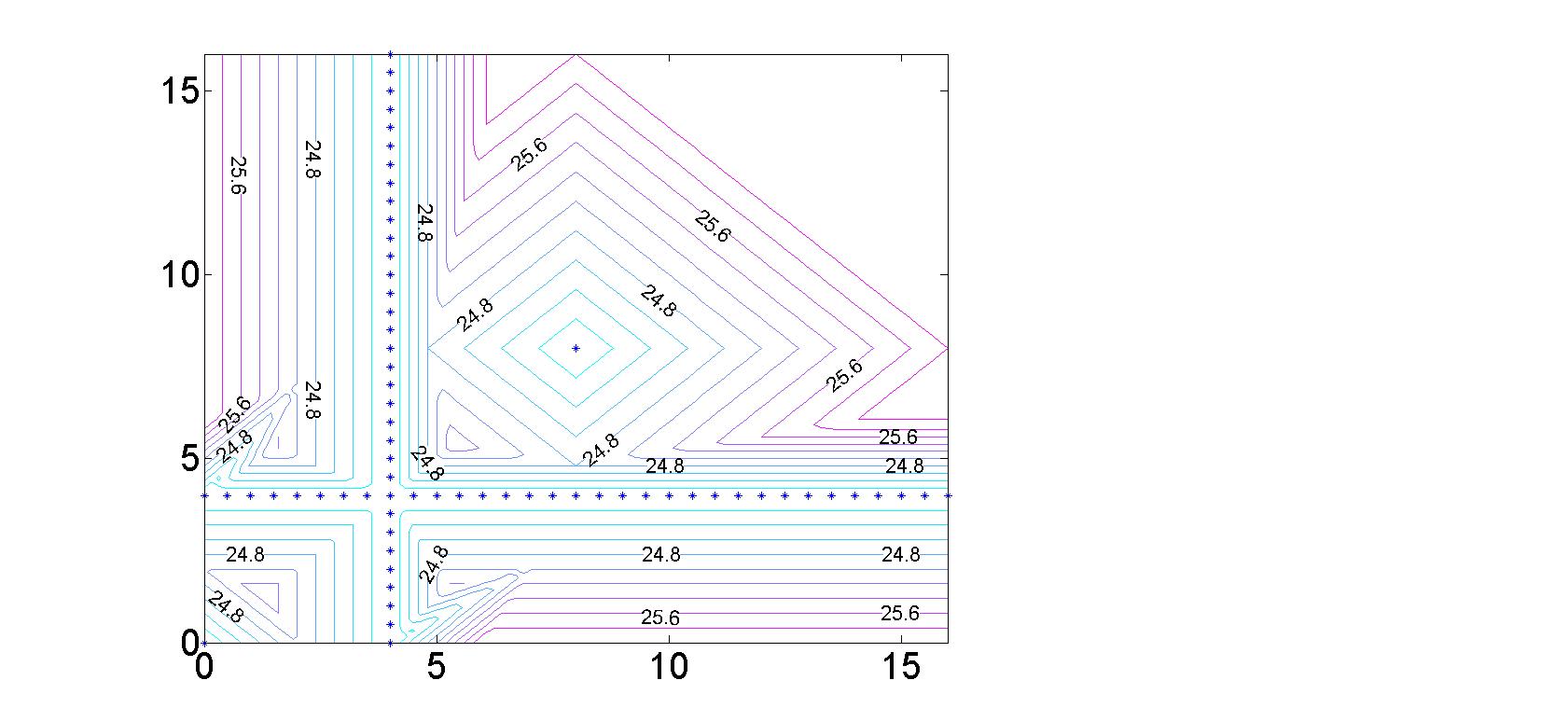}\caption{Contour plot of the values of the game $G_2^{and}$ for various combination of marker dropping times. The minimal value computed is $24$ for dropping times $(0,0), (8,8), (4,x)$ and $(x,4)$.}\label{2marksfigcontour}
\end{figure}
\begin{figure}
\hskip -3.cm\includegraphics[height=6cm, width=6cm]{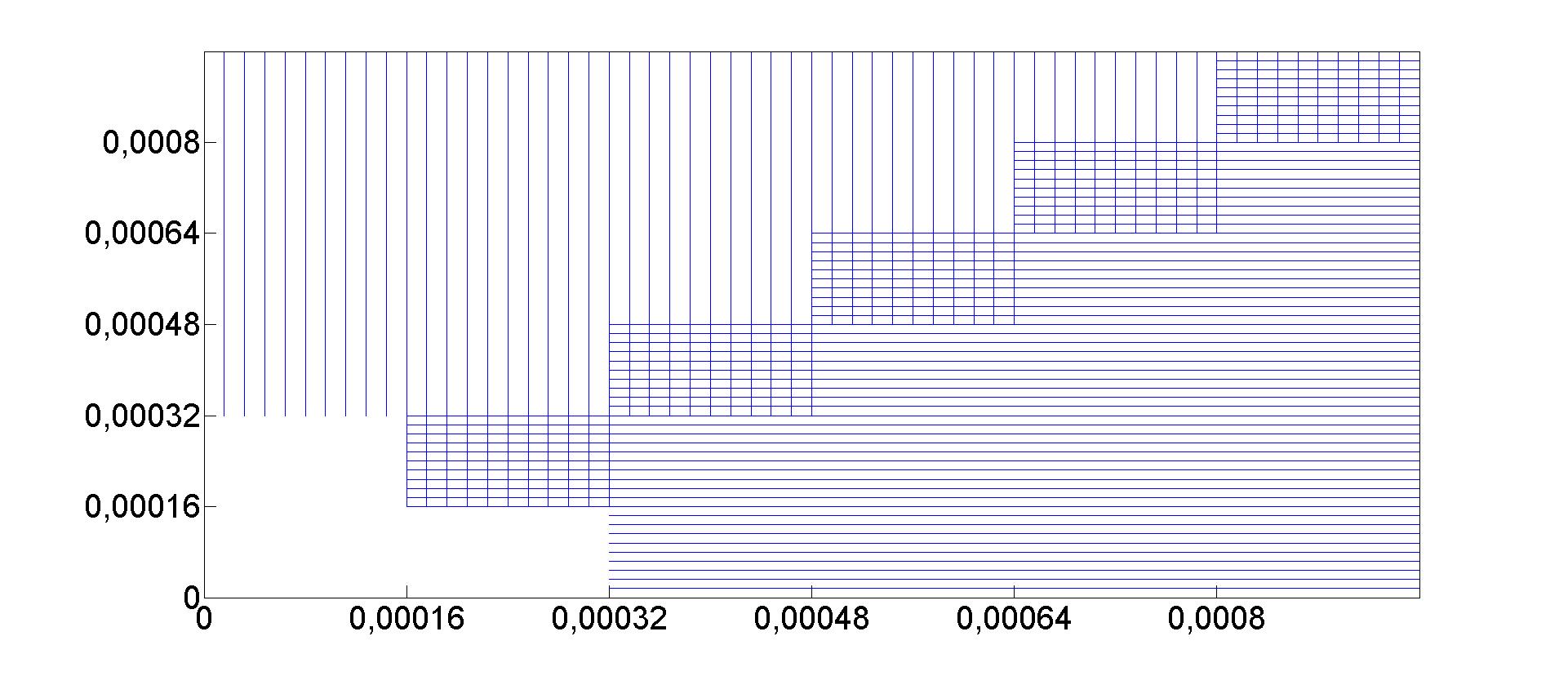}\includegraphics[height=6cm, width=6cm]{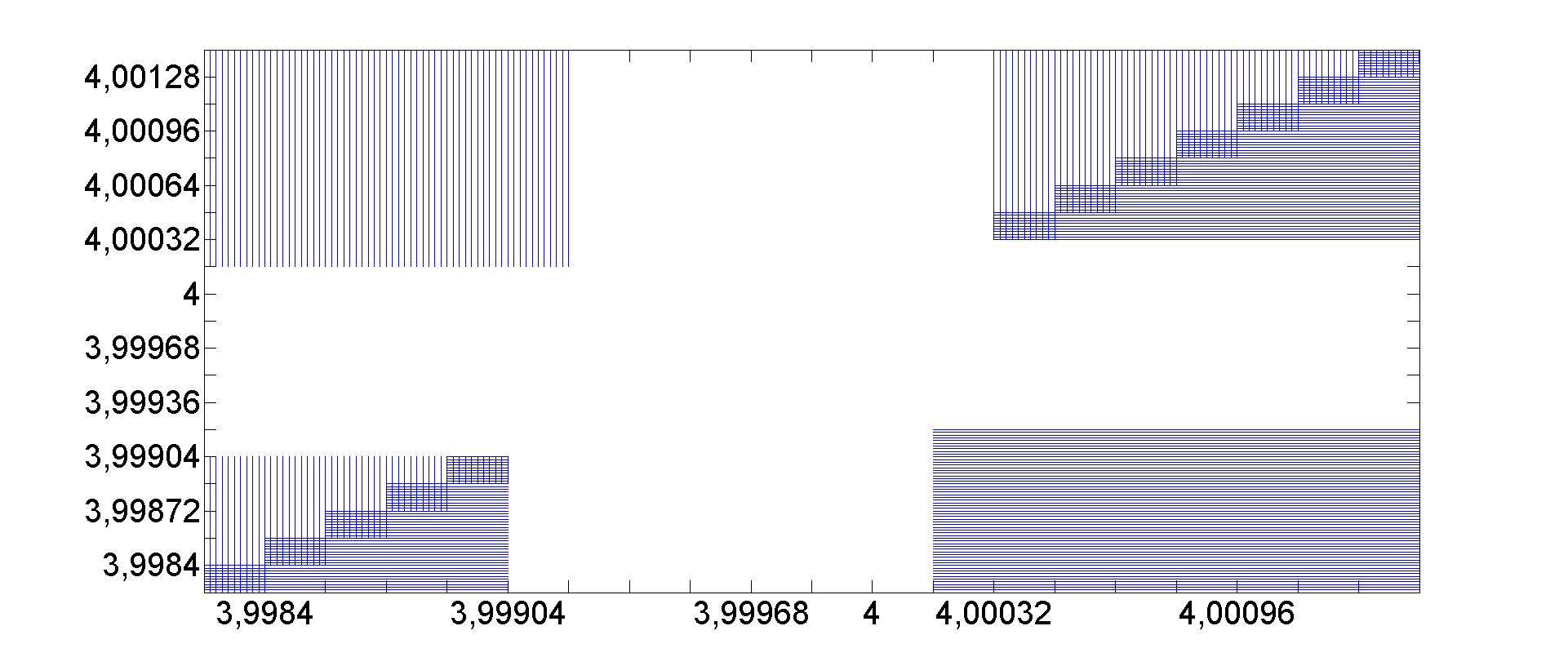}\includegraphics[height=6cm, width=6cm]{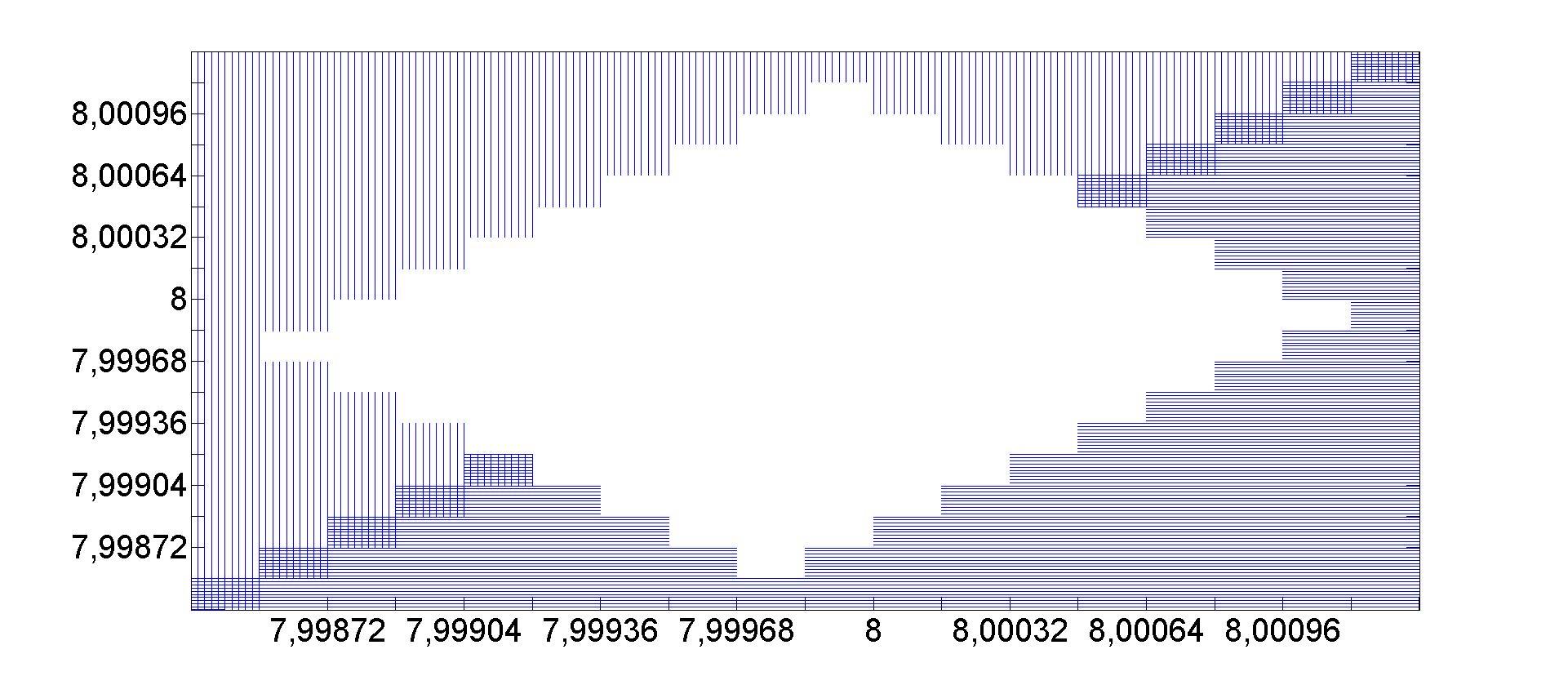}
\caption{}\label{fig:marker()}
\end{figure}


For the $G_2^{or}$ game, the minimal value computed is $20$ for dropping times $(8,8)$. The contour plot is displayed on Figure \ref{fig:giftor}. Interestingly, we observe that most of the dropping time values in the unit interval help the coordination.On the right of Figure \ref{fig:giftor} we see the excluded region by Theorem \ref{theo:two} that shows the optimal dropping times are in the region $[7.99856,8.00112]\times [7.99856,8.00112]$.

\newpage
\bibliographystyle{plain}
\bibliography{biblioRendezVous}
\end{document}